\documentclass[11pt]{article}
\usepackage[T1]{fontenc}
\usepackage{balance} 
\usepackage{authblk}
\usepackage[a4paper,margin=1in]{geometry}

\usepackage[numbers]{natbib}
\usepackage{soul}
\usepackage{url}
\usepackage{xspace}
\usepackage[hidelinks]{hyperref}
\usepackage[utf8]{inputenc}
\usepackage[small]{caption}
\usepackage{graphicx}
\usepackage{tabularx}
\usepackage{amsmath}
\usepackage{amsthm}
\usepackage{amsfonts,amssymb}
\usepackage{booktabs}
\usepackage{algorithm}
\usepackage{algorithmic}
\usepackage[switch]{lineno}
\usepackage{cleveref}
\usepackage{etoolbox} 
\AtEndEnvironment{proof}{\null\hfill\qed}
\usepackage{mathtools}
\usepackage{xcolor}
\usepackage[shortlabels,inline]{enumitem}
\usepackage{thmtools}
\usepackage{amsfonts}
\usepackage{tikz}
\usetikzlibrary{positioning, fit, backgrounds, calc}
\usetikzlibrary{arrows.meta}

 \newtheorem{theorem}{Theorem}[section]
 \newtheorem{lemma}[theorem]{Lemma}
 
 \newtheorem{claim}[theorem]{Claim}

 \newtheorem{proposition}[theorem]{Proposition}
 \newtheorem{corollary}[theorem]{Corollary}

 \theoremstyle{definition}

 \newtheorem{definition}[theorem]{Definition}

\newcommand{\R}{\mathbb{R}}

\newcommand{\PO}{\texttt{PO}}
\newcommand{\USW}{\texttt{USW}\xspace}
\newcommand{\NSW}{\texttt{NSW}\xspace}
\newcommand{\OPT}{\texttt{OPT}\xspace}
\newcommand{\maxUSW}{\texttt{maxUSW}\xspace}
\newcommand{\maxNSW}{\texttt{maxNSW}\xspace}
\newcommand{\Pos}{\texttt{Pos}}
\newcommand{\Neg}{\texttt{Neg}}
\newcommand{\LP}{\mathrm{LP}}
\newcommand{\MMS}{\texttt{MMS}\xspace}

\newcommand{\floor}[1]{\left\lfloor #1 \right \rfloor}
\newcommand{\SPOC}{\textrm{SPOC}\xspace}

\begin{document}
\title{Fair Allocation with Dichotomous Preferences over Item Complements}

%
%
\author[1]{Soumyajit Pyne\thanks{soumyajit.pyne@tifr.res.in}}
\author[2]{Vignesh Viswanathan\thanks{vviswanathan@umass.edu}}
\author[2]{Wenhao Yang\thanks{wenhaoyang@umass.edu}}
\author[2]{Yair Zick\thanks{yzick@umass.edu}}

\affil[1]{Tata Institute of Fundamental Research, Mumbai}
\affil[2]{University of Massachusetts, Amherst}
\maketitle              
\begin{abstract}
We study the problem of fair allocation when agents have dichotomous submodular preferences over groups of, rather than individual, items. 
These preferences naturally extend existing models, and have a natural interpretation in settings such as course allocation. 
We show that this seemingly minor extension results in computational intractability barriers: computing a utilitarian or Nash-optimal allocation with item complements is NP-hard. However, when agents' complementary bundles have size at most $k$, there exist poly-time algorithms that compute $\frac1k$-\maxUSW allocations and approximate Nash-welfare-maximizing allocations. While PO+EF-1 allocations may not exist in the complementary preference setting, it is possible to compute $\frac1{k+1}$-approximately \maxUSW EF-X allocations in polynomial time; furthermore, these allocations offer each agent $i$ a utility of at least $\frac{\MMS_i - k}{k+1 + \frac1n}$, where $\MMS_i$ is agent $i$'s maximin share guarantee. 

While the standard dichotomous binary submodular setting admits simple sequential techniques for computing fair and efficient allocations, our results for the complements setting require more complex tools from combinatorial optimization. 
\end{abstract}
\section{Introduction}
We study fair allocation with complementary preferences when agents have ``mostly'' dichotomous submodular preferences. Each agent has a list of complementary bundles of items; when each complementary bundle is `merged' into a single item, the agent's valuation function becomes submodular. 
We refer to this class of valuations as \emph{Submodular Preferences Over Complements} or \SPOC. 
These valuations are motivated by settings where agents require complementary resources to complete a task.

\paragraph{Course allocation problem.} The \emph{course allocation problem} focuses on the design of systems that assign seats in classes to students. This problem is naturally cast as an instance of the \emph{fair allocation} problem \cite{diaz2024deploying,bouveret2016handbook,Budish2011EF1,budish2017coursematch,soumalias2024courseallocation}. Students normally consider a variety of constraints when choosing classes. 
    One such constraint may be \emph{scheduling}: students may not take classes scheduled at the same time. Capacity constraints make some classes \emph{substitutes}: one class could replace another. 
    Not all constraints induce substitute item preferences. For example, a student may want to take classes A and B in the fall semester, so that they may take class C in the spring semester. 
    In this case, the two seats are \emph{complements} \cite{paesleme2017gross}. When items are substitutes, agents have \emph{submodular} valuations \cite{edmonds1970matroid,krause2014submodular,paesleme2017gross}: agents exhibit diminishing marginal returns as their bundles grow.

\paragraph{Allocation of multiple resource types.}
\cite{ghodsi2011dominant} study fair allocation in shared computing systems, where users need several resources together, such as CPU and memory, to run their tasks. Related works study fairness and efficiency in these settings~\cite{joe2013multiresource,parkes2015beyond}. 
Similarly, pharmaceutical testing \cite{martins2023minimizing} and semiconductor manufacturing \cite{bitar2016memetic} require large-scale assignments of multiple complementary resources to tests/machines. Here, we can think of the users as agents and discrete resource reservations as goods. The resources needed to run a task are complementary: tasks cannot run without all required resources. 
\SPOC preferences capture this behavior.

While fair allocation with submodular preferences is extensively studied (see \Cref{sec:related}), preferences over item complements are not as well understood. This raises a natural question: what guarantees can we offer when agent preferences are ``submodular in the large'' but exhibit \emph{some} item complements? 

	\subsection{Our contribution}\label{sec:contrib}
   We study binary \SPOC valuations. \SPOC preferences naturally extend binary submodular preferences, while highlighting the challenges introduced by item complements to both existence and computation of fair and efficient allocations. We divide our results into two parts (see \Cref{tab:summary-results}). The first focuses on computing and approximating maximum utilitarian and Nash welfare. 
    The second studies the tradeoff between fairness and efficiency, providing both positive guarantees and lower bounds.

    First, we establish the APX-hardness of computing a utilitarian or Nash-optimal allocation under binary \SPOC preferences (\Cref{thm:nphard-shared-synergies}). We show that computing a better than $\frac{3555}{3556}$ approximation to \maxUSW is NP-hard, even for binary $\SPOC_2$ agents. Our proof uses a reduction from an APX-hard version of 2P2N-3SAT~\cite{berman20032P2N}.  
    The same reduction also shows that computing any constant factor approximation to the \maxNSW is NP-hard.
    
    We complement these results with two poly-time algorithms for computing approximately optimal allocations for the utilitarian (\Cref{thm:matroid-l}) and Nash (\Cref{thm:5-appx}) welfare objectives. 
    A generalization of \Cref{alg:iterative-matroid-basis} computes a $\frac1k$-\maxUSW allocation in polynomial time for agents whose complementary bundles are of size $\le k$. We utilize a linear programming relaxation of the \maxUSW problem with matroid constraints, whose optimal value upper-bounds \maxUSW. We then use iterative rounding to construct (in polynomial time) an integral solution with value at least $\frac{1}{k}$ times the LP optimum, where $k$ is the size of the complements, yielding a $\frac{1}{k}$-approximation to \maxUSW.

    While computing a constant factor approximation to the Nash welfare is computationally intractable (\Cref{corr:nash-hard}), \Cref{alg:nash-welfare-local-search} efficiently computes an approximate \maxNSW allocation where each agent's utility is augmented by a small constant $\beta$.  
    We use a variant of the algorithm in~\cite{Barman2021approxmnw}.
    Our analysis uses a factor-revealing LP to obtain an approximation guarantee with both multiplicative and additive terms. Our analysis also improves the approximation guarantee for binary XOS valuations in~\cite{Barman2021approxmnw} from $288$ to $6$.
    
    Next, we explore fairness objectives under binary \SPOC preferences. 
    We show that there may not exist allocations that guarantee each agent any approximation of their maximin share, nor is a Pareto-optimal and EF-X allocation guaranteed to exist. 
    However, we offer positive results for agents with complements of size exactly $k$. \Cref{alg:clean-efx} computes a $\frac1{k+1}$-\maxUSW EF-X allocation, which also offers each agent $i$ a utility of at least $\frac{\MMS_i - k}{k+1+\frac1n}$, where $\MMS_i$ is agent $i$'s maximin share guarantee.
    \Cref{alg:clean-efx} repeatedly selects an agent with the lowest utility among those who envy the unallocated item pool. It then replaces this agent's current bundle with a suitable bundle from the pool, returning the old bundle to the pool. Finally, we show that we can complete the partial allocation returned by \Cref{alg:clean-efx} while preserving all of our guarantees.
\begin{table*}[t]
\centering
\small
\setlength{\tabcolsep}{4pt}
\renewcommand{\arraystretch}{1.08}
\caption{Summary of results for binary $\SPOC_k$ valuations. All algorithms run in
polynomial time.}
\label{tab:summary-results}
\begin{tabularx}{0.9\textwidth}{@{}l
    >{\raggedright\arraybackslash}p{0.30\textwidth}
    >{\raggedright\arraybackslash}X@{}}
\toprule
\textbf{Objective} & \textbf{Lower bound} & \textbf{Upper bound} \\
\midrule
\maxUSW
& APX-hard, even for $k=2$
& $\frac{1}{k}$-approximation \\

\maxNSW
& No constant factor multiplicative approximation, even for $k=2$
& $(\alpha,\beta)$-approximation for $\alpha\ge k^2+1$,
  $(\alpha-1)\beta\ge k^2$ \\

EF-X and efficiency
& EF-1 and PO need not coexist
& EF-X for complements of size $k$, with
  $\USW(X)\ge\frac{1}{k+1}\maxUSW$ and
  $v_i(X_i)\ge\frac{\MMS_i-k}{k+1+1/n}$ for all $i$ \\
\bottomrule
\end{tabularx}
\end{table*}

	\subsection{Related work}\label{sec:related}
    Dichotomous agent preferences are extensively studied. 
    \cite{Babaioff2021Dichotomous,benabbou2019group,benabbou2021MRF} examine binary submodular valuations, presenting algorithms for computing fair and efficient allocations. 
    Several works present \emph{truthful} mechanisms for computing fair and efficient allocations for binary additive \cite{halpern2020binaryadditive} and submodular \cite{Babaioff2021Dichotomous,viswanathan2022yankee} valuations; \cite{barman2022groupstrategyproof} extend these results to \emph{group} strategyproof mechanisms. 
    Binary submodular valuations admit poly-time methods for computing fair-share guarantees \cite{Barman2021MRFMaxmin}, as well as truthful poly-time computation of fair and efficient allocations \cite{montanari2025wefMRF,viswanathan2025GYS}. 

    Literature on non-submodular agent preferences in fair division is more limited. \cite{Lipton2004EF1} present a poly-time algorithm for computing fair allocations under general valuations.
    \cite{ghodsi2022BeyondAdditive} focus on the maximin share guarantee; \cite{akrami2025eefx} show that fair allocations exist for general valuation classes; \cite{barman2024nearly} show that EQ-X allocations exist for general monotone valuations and can be computed in pseudo-polynomial time, while other works focus on finding fair allocations for subadditive and XOS valuations \cite{chaudhury2021subadditive,feige2025subadditive}.

    Course allocation mechanisms elicit complementary preferences in practice. \cite{budish2017coursematch} and \cite{soumalias2024courseallocation} elicit bundle preferences as well as individual course preferences from students.  
    \cite{Barman2021approxmnw} use a method similar to \Cref{alg:nash-welfare-local-search} to efficiently compute an allocation that offers a $288$ approximation to the maximum Nash welfare when agents have binary XOS valuations.
	\section{Preliminaries}
	We use $[t]$ to denote the set $\{1, 2,\dots , t \}$.
	Given sets $A,B$, we replace $A\setminus B$ (resp. $A\cup B$) with $A-B$ (resp. $A+B$); when $B$ is a single element, i.e., $B = \{g\}$, we replace $A \setminus \{g\}$ (resp. $A \cup \{g\}$) with $A - g$ (resp. $A + g$).
	We have a set of $n$ \emph{agents} $N = [n]$ and a set of $m$ \emph{items} $G = \{g_1, \dots, g_m\}$.
	Each agent $i \in N$ has a \emph{valuation function} $v_i: 2^G \to \R_+$; given a bundle $S \subseteq G$, $v_i(S)$ denotes how much agent $i$ values the bundle $S \subseteq G$.
	Let $\Delta_i(S, T) \triangleq v_i(S+T) - v_i(S)$ denote the \emph{marginal utility} agent $i\in N$ gets from the bundle $T\subseteq G$, given that they have the bundle $S\subseteq G$. 
    When $T = \{g\}$, we write $\Delta_i(S,g)$ instead of $\Delta_i(S,\{g\})$. 
	
	\emph{Binary submodular} valuations, also known as \emph{matroid rank functions} or MRFs \cite{oxley2011matroids}, are a key component of our work.
	An MRF valuation function $v_i$ satisfies the following three properties: 
	\begin{enumerate*}[label={(\alph*)}]
		\item $v_i(\emptyset) = 0$
		\item for any bundle $S \subseteq G$ and item $g \in G$, $v_i(S+g) - v_i(S) \in \{0, 1\}$ and
		\item $v_i$ is submodular: for any sets $S \subseteq T \subseteq G$ and any item $g \in G$, $\Delta_i(S, g) \geq \Delta_i(T, g)$.
	\end{enumerate*}
	\paragraph{Submodular Preferences Over Complements.} 
    We investigate the effect of \emph{complementary} preferences (or preferences with item complements), under the following model. Every agent $i$ has a \emph{complements list} of disjoint item bundles $\mathcal L_i$ with the following property: for every bundle $S \subseteq G$, $v_i(S) = r_i(\{C \in \mathcal L_i:C \subseteq S\})$, where $r_i$ is a matroid rank function (MRF) over the bundles in $\mathcal L_i$. Informally, if we treat the bundles in $\mathcal L_i$ as individual items, the valuation $v_i$ becomes a matroid rank function. 
    We refer to agents whose valuations become MRFs once their complementary items have been merged as having \emph{binary submodular preferences over complements} (binary \SPOC). 
    When every element of $\mathcal L_i$ is of size $\le k$, we refer to agents' valuations as $\SPOC_k$.  
    
	We say that a binary $\SPOC_k$ agent $i$ is \emph{purely complementary} if for every bundle $C \in \mathcal L_i$, $|C| = k$.  
    Put simply, agent $i$ does not value any bundle of size $<k$, and only receives a positive utility from (some) $k$-sized bundles. 
	It is often no loss of generality to assume that $\SPOC_k$ agents are purely complementary. If $\Delta_i(S,T)=1$ for some bundle $T$ of size $<k$, we introduce $k - |T|$ dummy items that are only complements for agent $i$ and the bundle $T$.  

    When analyzing envy (see \Cref{sec:usw-efx}), introducing dummy items \textbf{is} a loss of generality; the EF-X guarantees we present only hold when all $\SPOC_k$ agents are purely complementary.

    Our goal is to allocate the items $G$ to the agents in some reasonable manner.
	An \emph{allocation} is a partition $X = (X_0; X_1, \dots, X_n)$ of the items $G$, where agent $i$ receives the bundle $X_i$, and $X_0$ is the set of unassigned items. 
	An allocation $X$ is \emph{clean} (or \emph{non-redundant}) \cite{benabbou2021MRF} if for any agent $i$ and any item $g$ in their bundle, $v_i(X_i - g) < v_i(X_i)$.
    We seek allocations satisfying certain \emph{justice criteria}. 
    \begin{description}[wide, leftmargin=0cm]
    \item[Envy:] an allocation $X$ is \emph{envy-free} (EF) \cite{foley1967resource} if no agent prefers another agent’s bundle to their own, i.e., for every two agents $i,j \in N$, $v_i(X_i) \ge v_i(X_j)$. 
	EF is not guaranteed when all goods are allocated; when there are two agents and only one item valued at $1$ by both, one agent will invariably envy the other.
    An allocation $X$ is EF-1 \cite{Budish2011EF1,Lipton2004EF1} if no agent envies another agent after dropping some good from the latter agent's bundle. 
    Formally, for any two agents $i,j\in N$ where $v_i(X_i) < v_i(X_j)$, there is some item $g$ in $X_j$ such that $v_i(X_i) \geq v_i(X_j -g)$.
	Similarly, an allocation $X$ is \emph{envy-free up to $\ell$ items} (EF-$\ell$) if for any two agents $i,j\in N$ where $v_i(X_i) < v_i(X_j)$, there is some set of $\le \ell$ items $S \subseteq X_j$ such that $v_i(X_i) \geq v_i(X_j \setminus S)$. 
    An allocation is \emph{envy-free up to any item} (EF-X) \cite{Caragiannis2016MNW} if no agent envies another agent after dropping any good from the latter agent's bundle. Formally, for any agents $i,j\in N$ where $v_i(X_i)<v_i(X_j)$, and for every item $g\in X_j$, $v_i(X_i)\ge v_i(X_j-g)$. 
    \item[Maximin Share Value:] 
    Let $\Pi_n(G)$ denote the set of all partitions of $G$ into $n$ bundles. The \emph{maximin share} (\MMS) of agent $i$ is defined as
    $
        \MMS_i=\max_{(X_1,\ldots,X_n)\in\Pi_n(G)}
        \min_{j\in[n]} v_i(X_j).
    $
    Intuitively, $\MMS_i$ represents agent $i$'s utility when they partition the goods into $n$ bundles and then
    receive their least valuable bundle.
    \end{description}
    We seek allocations satisfying some efficiency metric as well; common metrics include:
    \begin{description}[align=left,wide]
    \item[Pareto Optimality:]
	An allocation $X$ is \emph{Pareto optimal} (\PO) if there does not exist any allocation $Y$ such that for all agents $i \in N$, $v_i(Y_i) \geq v_i(X_i)$ and there exists at least one agent $j \in N$ such that $v_j(Y_j) > v_j(X_j)$.
	\item[Utilitarian Social Welfare:]
	The \emph{utilitarian social welfare} (\USW) of an allocation is $\USW(X) = \sum_{i\in N} v_i(X_i)$. We refer to an allocation that maximizes \USW as a \maxUSW allocation. 
	\item[Nash Social Welfare:]
	The \emph{Nash social welfare} (\NSW) of an allocation is the geometric mean of agent utilities under that allocation, i.e., $\NSW(X) = \left(\prod_{i\in N} v_i(X_i)\right)^{\frac{1}{n}}$. It is possible that in every allocation $X$, there is some agent $i \in N$ such that $v_i(X_i) = 0$. 
    In such cases, a Nash-optimal allocation first minimizes the number of $0$ utility agents, and then maximizes the product of utilities among the non-zero utility agents \cite{Caragiannis2016MNW}. Such allocations are called $\maxNSW$ allocations.  
    \end{description}
    Since computing \maxUSW/\maxNSW allocations in our setting is computationally intractable, we study allocations that approximate these objectives.
    Given some $\alpha\in [0,1]$, an allocation $X$ is $\alpha$-\maxUSW (\maxNSW) if $\USW(X)$ ($\NSW(X)$) is at least $\alpha$ times the \USW (\NSW) of the \maxUSW (\maxNSW) allocation.

\SPOC preferences present existential, as well as computational, barriers: there exists an instance with binary $\SPOC_2$ agents for which no allocation satisfies any reasonable fairness/efficiency property.
	\begin{proposition}\label{thm:ef1po}
		There exists an instance with binary $\SPOC_2$ agents where
        \begin{enumerate*}
            \item No EF-1 allocation is PO. 
            \item No allocation offers all agents any constant fraction of their maximin share. 
        \end{enumerate*}
	\end{proposition}
	\begin{proof}
		Consider the following instance. 
        There are four items $\{g_1, g_2, g_3, g_4\}$ and two agents $A, B$. Agent $A$ is additive over the pairs $\{g_1, g_2\}$ and $\{g_3, g_4\}$. Agent $B$ is additive over the pairs $\{g_1, g_3\}$ and $\{g_2, g_4\}$.
		The only PO allocations assign all items to the same agent. However, those allocations are not EF-1.

        In the instance above, the maximin share guarantee cannot be approximated to any positive factor: $\MMS_A = \MMS_B = 1$, but any allocation that offers a positive utility to one agent offers a utility of $0$ to the other.  
	\end{proof}
	\Cref{thm:ef1po} further implies that achieving EF-1 in combination with either maximum social welfare (\maxUSW) or maximum Nash welfare (\maxNSW) is impossible in the \SPOC agent setting; this is in stark contrast to stronger fairness/efficiency guarantees when agents have additive or MRF preferences \cite{Caragiannis2016MNW,Babaioff2021Dichotomous,benabbou2021MRF}. 

\section{Welfare-Maximizing Allocations for Binary SPOC agents}\label{sec:complexity}
We first show the APX-hardness of computing a \maxUSW or \maxNSW allocation for binary $\SPOC_2$ agents.
This is true even when every agent has maximum utility $1$ and all agents share a complements master list. 
\begin{restatable}{theorem}{thmnphardsharedsynergies}\label{thm:nphard-shared-synergies}
For every fixed $\alpha>\frac{3555}{3556}$, there is no polynomial-time algorithm that computes an $\alpha$-\maxUSW allocation, unless $\mathrm{P}=\mathrm{NP}$.
\end{restatable}
To prove \Cref{thm:nphard-shared-synergies}, we present a reduction from the \emph{2P2N-3SAT problem}, a variant of 3SAT in which each variable appears exactly four times, twice positively and twice negatively. Specifically, we use the following variant.
\begin{lemma}[Berman, Karpinski, and Scott~\cite{berman20032P2N}]
\label{lem:bks-gap}
For every sufficiently small constant $\eta>0$, it is NP-hard to
distinguish between the following two types of 2P2N-3SAT instances:
\[
\begin{array}{lll}
\textnormal{\textsc{Yes}:}
& s^* \geq (1016-\eta)t,
\textnormal{\textsc{   No}:}
& s^* \leq (1015+\eta)t,
\end{array}
\]
where the formula has exactly $1016t$ clauses and $s^*$ denotes the
maximum number of simultaneously satisfiable clauses.
\end{lemma}

\begin{proof}[Proof of \Cref{thm:nphard-shared-synergies}]
Given an instance $\mathcal{F}$ of 2P2N-3SAT with clauses $c_1,\ldots,c_k$, 
we construct a corresponding allocation instance as follows (see \Cref{fig:reduction2p2n3sat} for an illustration). 
 For each variable $x$, we introduce four variable items, $\Pos(x)_1$, $\Pos(x)_2$, $\Neg(x)_1$, and $\Neg(x)_2$, along with a variable agent $X$ who wants at most one of the complementary item pairs $\{\Pos(x)_1, \Pos(x)_2\}$ or $\{\Neg(x)_1, \Neg(x)_2\}$. For each clause $c_j = \ell^j_1 \lor \ell^j_2 \lor \ell^j_3$, we create a clause agent $C_j$ that wants at most one variable item matching a literal in its clause. For instance, in \Cref{fig:reduction2p2n3sat}, $c_4 = \neg x \vee \neg y \vee z$, so the clause agent $C_4$ wants exactly one item from the set $\Neg(x)_1$, $\Neg(x)_2$, $\Neg(y)_1$, $\Neg(y)_2$, $\Pos(z)_1$, and $\Pos(z)_2$. 
The maximum utility of any agent is $1$, so the total possible \USW attainable is upper bounded by the number of agents. Call this reduced instance $\mathcal{R}(\mathcal{F})$. Consider a 2P2N-3SAT instance from \Cref{lem:bks-gap}. It contains $1016t$ clauses. 
Since every variable occurs four times and every clause has three literals, the number of variables is $\frac34\times 1016t = 762t$.
\begin{description}[wide]
\item[Case 1:]
Suppose the number of satisfied clauses $s^*$ is greater than $(1016 - \eta)t$; consider an assignment where at least $(1016-\eta)t$ clauses are satisfied. 
For this assignment, for each variable $x$, we assign its corresponding agent $X$ the items $\{\Neg(x)_1, \Neg(x)_2\}$ if $x$ is true, and $\{\Pos(x)_1, \Pos(x)_2\}$ otherwise. 
Therefore, all variable agents get a utility of 1.

If a clause $c_j$ is satisfied by a variable $x$ set to `true', then we assign the clause agent $C_j$ an available item from $\{\Pos(x)_1,\Pos(x)_2\}$; if the literal is satisfied by setting $x$ to `false', we assign $C_j$ an item in $\{\Neg(x)_1,\Neg(x)_2\}$. This assignment agrees with how items are assigned to the corresponding variable agents.  
In this assignment, variable agents and all clause agents whose clauses are satisfied  receive a utility of $1$ for a total utility of at least $762t+ (1016-\eta)t = (1778-\eta)t$; thus, the \maxUSW is at least $(1778-\eta)t$.

\item[Case 2:] suppose next that $s^*\le(1015 + \eta)t$, i.e., no truth assignment for the instance $\mathcal{F}$ satisfies more than $(1015 + \eta)t$ clauses. 
The total number of agents in $\mathcal{R}(\mathcal{F})$ is $762t + 1016t = 1778t$ and each agent receives a utility of at most $1$, so the $\maxUSW$ is at most $1778t$. Consider an allocation $A$ for $\mathcal{R}(\mathcal{F})$ where $z_V$ and $z_C$ respectively denote the number of zero utility variable agents and clause agents. We construct a truth assignment for $\mathcal F$ from this allocation by considering the variable agents' utilities. 

Consider a variable agent $X$ whose utility is $1$. If $X$ received $\{\Pos(x)_1,\Pos(x)_2\}$ we set its corresponding variable $x$ to `false'; if it received $\{\Neg(x)_1,\Neg(x)_2\}$, we assign $x$ to be `true'.
If $X$ has a utility of $0$, we set the variable $x$ so that it satisfies at least two of the four clauses it appears in. 
    
 Under this assignment, the total number of unsatisfied clauses is at most $2z_V + z_C\le2(z_V + z_C)$. Note that $\USW(A) = 1778t - (z_V + z_C)$. 
 We know that in the 2P2N-SAT instance $\mathcal{F}$, the total number of unsatisfied clauses is at least $1016t - (1015+\eta)t = (1-\eta)t$. 
 Therefore,
 \[
 \USW(A) = 1778t - (z_V + z_C) \le 1778t - (1-\eta)\frac{t}{2} = \left(1777.5 + \frac{\eta}{2}\right)t.
 \]
 Therefore, the \maxUSW of the corresponding fair allocation instance is at most $\left(1777.5 + \frac{\eta}{2}\right)t$.
 \end{description}

Suppose that there is a polynomial-time algorithm that computes an $\alpha$-\maxUSW allocation where $\alpha > (1777.5/1778) = (3555/3556)$. There exists some sufficiently small $\eta>0$ such that $\alpha(1778-\eta)>1777.5 + \frac\eta2$. 
From Case 1, we know that if $s^*\ge(1016-\eta)t$, then \maxUSW is at least $(1778 - \eta)t$. Therefore, $\USW(A) \ge \alpha\maxUSW\ge\alpha(1778 - \eta)t>(1777.5+\eta/2)t$.
From Case 2, we know that if $s^*\le(1015+\eta)t$, then \maxUSW is at most $(1777.5+\eta/2)t$. Therefore, comparing $\USW(A)$ with the threshold $(1777.5 + \eta/2)t$, we will be able to distinguish between the YES or NO instances in $\mathcal{F}$, a contradiction.
\end{proof}

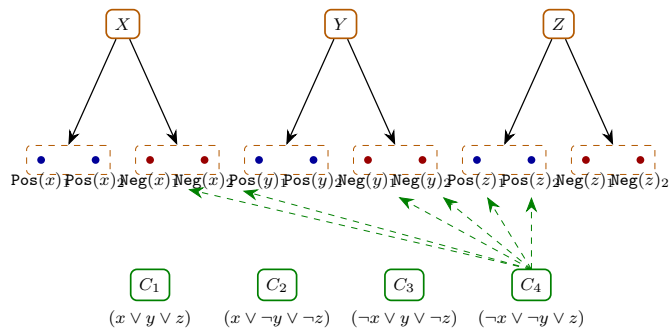
\begin{figure}[ht]
\centering
\scalebox{0.80}{%
    \begin{tikzpicture}[
	node distance = .9cm,
    item/.style={draw, circle, inner sep=1pt,
                 font=\scriptsize, thick},
    positem/.style ={item, fill = blue!60!black,draw=blue!60!black},
    negitem/.style ={item,  fill = red!60!black, draw=red!60!black},
    varagent/.style   ={inner sep = 4pt,draw, rectangle, rounded corners=3pt, thick,
                        draw=orange!70!black, font=\scriptsize},
    clauseagent/.style={inner sep = 4pt,draw, rectangle, rounded corners=3pt, thick,
                        draw=green!50!black, font=\scriptsize,outer sep=0pt},
    syngroup/.style ={rounded corners=2pt, draw=orange!70!black, dashed,
                     inner sep=5pt},
    synedge/.style  ={semithick, black,-{Stealth[length=2.5mm]}},
    cedge/.style    ={black!50!green,dashed,-{Stealth[length=2.5mm]}}
]

\node[positem] (px1)[label = below:{\scriptsize$\Pos(x)_1$}] {};
\node[positem] (px2) [right of =px1,label = below:{\scriptsize$\Pos(x)_2$}]{};
\node[negitem] (nx1) [right of = px2,label = {[name=nx1label] below:{\scriptsize$\Neg(x)_1$}}] {};
\node[negitem] (nx2) [right of = nx1,label ={[name=nx2label] below:{\scriptsize$\Neg(x)_2$}}] {};
\node[positem] (py1) [right of =nx2,label = {below:{\scriptsize$\Pos(y)_1$}}] {};
\node[positem] (py2) [right of = py1,label = below:{\scriptsize$\Pos(y)_2$}] {};
\node[negitem] (ny1) [right of = py2,label = {[name=ny1label] below:{\scriptsize$\Neg(y)_1$}}] {};
\node[negitem] (ny2) [right of = ny1,label ={[name=ny2label] below:{\scriptsize$\Neg(y)_2$}}] {};
\node[positem] (pz1) [right of = ny2,label ={[name=pz1label]below:{\scriptsize$\Pos(z)_1$}}] {};
\node[positem] (pz2) [right of = pz1,label = {[name=pz2label] below:{\scriptsize$\Pos(z)_2$}}] {};
\node[negitem] (nz1) [right of = pz2,label = below:{\scriptsize$\Neg(z)_1$}] {};
\node[negitem] (nz2) [right of = nz1,label = below:{\scriptsize$\Neg(z)_2$}] {};

\begin{scope}[on background layer]
	\node[syngroup, fit=(px1)(px2)] (sxp) {};
	\node[syngroup, fit=(nx1)(nx2)] (sxn) {};
	\node[syngroup, fit=(py1)(py2)] (syp) {};
	\node[syngroup, fit=(ny1)(ny2)] (syn) {};
	\node[syngroup, fit=(pz1)(pz2)] (szp) {};
	\node[syngroup, fit=(nz1)(nz2)] (szn) {};
\end{scope}

\node[varagent] (X) at ($(sxp.north)!0.5!(sxn.north) + (0, 2)$)  {$X$};
\node[varagent] (Y) at ($(syp.north)!0.5!(syn.north) + (0, 2)$) {$Y$};
\node[varagent] (Z) at ($(szp.north)!0.5!(szn.north) + (0, 2)$) {$Z$};

\draw[synedge] (X) -- (sxp.north);
\draw[synedge] (X) -- (sxn.north);
\draw[synedge] (Y) -- (syp.north);
\draw[synedge] (Y) -- (syn.north);
\draw[synedge] (Z) -- (szp.north);
\draw[synedge] (Z) -- (szn.north);

\coordinate (clauseL) at ($(nx1.south) + (0,-2.0)$);
\coordinate (clauseR) at ($(pz2.south) + (0,-2.0)$);

\foreach \i/\t in {1/0, 2/1, 3/2, 4/3}
\node[clauseagent] (C\i) at ($(clauseL)!\t/3!(clauseR)$) {$C_\i$};

\node[below=1pt of C1, font=\scriptsize] {$(x \lor y \lor z)$};
\node[below=1pt of C2, font=\scriptsize] {$(x \lor \neg y \lor \neg z)$};
\node[below=1pt of C3, font=\scriptsize] {$(\neg x \lor y \lor \neg z)$};
\node[below=1pt of C4, font=\scriptsize] {$(\neg x \lor \neg y \lor z)$};

\foreach \i in {nx1label, nx2label, ny1label, ny2label, pz1label, pz2label}
    \draw[cedge] (C4.north) -- (\i);
\end{tikzpicture}}
\caption{The reduction from 2P2N-3SAT to an instance of the fair allocation with \SPOC agents problem. Each variable agent wants exactly one complementary item pair of either the positive or negative items. The clause agents want exactly one of the items that correspond to the literals in their clause. }\label{fig:reduction2p2n3sat}
\end{figure}

\Cref{thm:nphard-shared-synergies} shows that, unless $\mathrm{P}=\mathrm{NP}$, no polynomial-time algorithm can compute an $\alpha$-\maxUSW allocation for any $\alpha>\frac{3555}{3556}$, even for binary $\SPOC_2$ agents. 
In particular, the problem admits no PTAS. The same reduction also implies the following corollary.

\begin{corollary}\label{corr:nash-hard}
    Suppose that all agents have binary $\SPOC_2$ preferences. Determining whether all agents receive a utility of $1$ is NP-hard, even when all agents can receive a utility of at most $1$.
\end{corollary}
\Cref{corr:nash-hard} implies that computing an $\alpha$-\maxNSW allocation in the \SPOC setting is computationally intractable for any constant $\alpha>0$. 
\subsection{Approximating the Utilitarian Welfare for SPOC agents}\label{sec:approx-usw}
We establish the poly-time computability of a $\frac{1}{k}$-\maxUSW allocation for binary $\SPOC_k$ agents. 

\begin{theorem}\label{thm:matroid-l}
There exists a polynomial-time algorithm that computes a $\tfrac{1}{k}$-\maxUSW allocation for binary $\SPOC_k$ agents.
\end{theorem}
We prove \Cref{thm:matroid-l} for the case $k=2$, corresponding to complementary item pairs. The proof for the general case is provided in the appendix. 
\Cref{thm:matroid-l} relies on matroid theory \cite{oxley2011matroids}. 
\begin{definition}[Matroid]
A \emph{matroid} is a pair $M = (S,\mathcal{I})$, where $S$ is a finite ground set and $\mathcal{I} \subseteq 2^S$ is a family of subsets, called \emph{independent sets}, satisfying the following properties:
\begin{enumerate*}[label = \textbf{\alph*.}]
    \item (\textbf{Non-emptiness}) $\emptyset \in \mathcal{I}$.
    \item (\textbf{Hereditary property}) If $I \in \mathcal{I}$ and $J \subseteq I$, then $J \in \mathcal{I}$.
    \item (\textbf{Exchange property}) If $I,J \in \mathcal{I}$ and $|I| < |J|$, then there exists an element $e \in J \setminus I$ such that $I \cup \{e\} \in \mathcal{I}$.
\end{enumerate*}
\end{definition}
A matroid structure makes hard problems tractable. 
For example, the matroid structure of acyclic subgraphs allows us to efficiently compute a minimal cost spanning tree for any given graph \cite{oxley2011matroids}.   
Given a matroid $M = (S,\mathcal I)$, its \emph{matroid rank function} (MRF) is
$r : 2^S \to \mathbb{Z}_{\ge 0}$, where for every $T \subseteq S$, $r(T) = \max \{|I| \mid I \subseteq T,\ I \in \mathcal{I}\}.$
Every MRF is binary submodular, and a binary submodular function is the rank function of some matroid \cite{oxley2011matroids}.

Given a matroid $M = (S,\mathcal{I})$, and an element $e \in S$, the \emph{deletion} of $e$ from $M$, denoted by $M \setminus e$, is the matroid defined on the ground set $S \setminus \{e\}$ with independent sets
$
\mathcal{I}(M \setminus e) = \{ I \subseteq S \setminus \{e\} \mid I \in \mathcal{I} \}
$.
Let $M = (S,\mathcal{I})$ be a matroid and let $e \in S$ be an element in the ground set. 
The \emph{contraction} of $e$ in $M$, denoted by $M / e$, is the matroid defined on the ground set $S \setminus \{e\}$. If $\{e\}$ is independent, then the independent sets of $M/e$ are
$
\mathcal{I}(M / e) = \{ I \subseteq S \setminus \{e\} \mid I \cup \{e\}\in \mathcal{I} \}. 
$
Otherwise, $\mathcal{I}(M / e) =\mathcal{I}$. When $\{e\}$ is not independent, $M\setminus e = M/e$.
\paragraph{Matroid formulation.}
We define a matroid $M=(S,\mathcal{I})$, whose ground set is $S = \{(i,\{g,g'\})\mid \{g,g'\}\in \mathcal L_i\}$. Intuitively, a tuple $(i,\{g,g'\})$ indicates that the complementary item pair $\{g,g'\}$ is assigned to agent $i$. Given a set $T\subseteq S$, for each $i\in N$, let $T_i$ be the set of all items appearing in tuples associated with agent $i$ in $T$. In other words, an item $g$ belongs to $T_i$ if $T$ contains a tuple $(i,\{g_1,g_2\})$ with $g=g_1$ or $g=g_2$. Thus, $T_i$ is the set of items allocated to agent $i$ under $T$. A set $T\subseteq S$ is independent, i.e., $T\in\mathcal{I}$, if and only if $T_i$ is independent for agent $i$ for every $i\in N$. Therefore, for any set $T\subseteq S$, the rank of $T$ is
$r(T) = \sum_i v_i(T_i)$.
\paragraph{Linear program formulation.}
Given an item $g\in G$, we define $S_g = \{(i,\{g_1,g_2\})\in S\mid g \in \{g_1,g_2\}\}$.
With this definition, we define the \emph{linear program} $\LP(M)$:
\begin{align}
\max \quad & \sum_{e \in S} x_e \\
\text{s.t.} \quad 
& \sum_{e \in T} x_e \le r(T) & \forall\, T \subseteq S \label{rank-cons}\\
&\sum_{e \in S_g} x_e \le 1 & \forall\, g \in G \label{matching-cons}\\
& 0\le  x_e \le 1 & \forall\, e \in S \label{boundary-cons}
\end{align}
We observe that there is a 1-1 correspondence between any feasible integral solution $\vec x$ to $\LP(M)$ and an allocation $X$ whose $\USW$ is the objective value of $\vec x$. 
The proof of \Cref{lm:lp-1} is in the appendix.
\begin{restatable}{lemma}{lemLP}\label{lm:lp-1}
For any clean allocation $X=(X_1,\ldots,X_n)$, there exists a feasible solution $\vec x$ to $\LP(M)$ such that $\USW(X) = \sum_{e \in S}x_e$. 
Conversely, any integral feasible solution $\vec x$ to $\LP(M)$ corresponds to a clean allocation $X$ where $\USW(X) = \sum_{e\in S}x_e$.
\end{restatable}
We use \emph{iterative rounding} to find an integral solution to $\LP(M)$ whose value is at least $\frac12$ of the optimal value.
We say that a solution $\vec x$ to $\LP(M)$ is an \emph{extreme point solution} if it is a vertex of the feasible polytope induced by the LP constraints. 

\begin{algorithm}[htbp]
\caption{Iterative Rounding}
\label{alg:iterative-matroid-basis}
\begin{algorithmic}[1]
\STATE Initialize $T \gets \emptyset$
\WHILE{$S\neq\emptyset$}
    \STATE Find an optimal extreme point solution $\vec x$ to $\mathrm{LP}(M)$
    \IF{ $\exists e\in S$ with $x_e = 0$}
    \STATE $M \gets M \setminus e$
    \ELSIF{there exists an element $e$ such that $x_e \ge \frac{1}{2}$}
        \STATE Update $T \gets T \cup \{e\}$
        \STATE Delete every remaining edge sharing a good with $e$
        \STATE Contract $e$ in $M$, i.e., $M \gets M / e$
    \ENDIF
\ENDWHILE
\RETURN $T$
\end{algorithmic}
\end{algorithm}

\begin{restatable}{lemma}{lemmaextremepoint}\label{lm:extreme-point}
Let $\vec x$ be an extreme point solution obtained during an iteration of Algorithm~\ref{alg:iterative-matroid-basis}. 
Then there exists an element $e \in S$ such that either $x_e = 0$ or $x_e \ge \tfrac{1}{2}$.
\end{restatable}
The proof of \Cref{lm:extreme-point} is in the appendix. Intuitively, we assume that for some extreme point solution $\vec x$ we have $0<x_e<\frac12$ for all $e\in S$, and analyze different conditions on the tight constraints. 

\begin{proof}[Proof of \Cref{thm:matroid-l} for $k=2$]
Suppose $\vec x$ is an extreme point solution of $\LP(M)$. 
By \Cref{lm:extreme-point}, at least one of the \textbf{If} or \textbf{ElseIf} conditions in the while loop of \Cref{alg:iterative-matroid-basis} is satisfied in every iteration. 
If there exists an element $e^* \in S$ with $x_{e^*} = 0$, then deleting $e^*$ from the matroid $M$ does not change the objective value of the LP defined over $M \setminus e^*$ (intuitively, if $e^* = (i,\{g,g'\})$ then agent $i$ does not want the item pair $\{g,g'\}$ nor will they ever want it given their assignment). Now consider the case where there exists an element $e^* = (i,\{g,g'\}) \in S$ such that $x_{e^*} \ge \tfrac{1}{2}$. 

We construct an alternative solution $\vec y$ for $\LP(M)$. 
Let $\vec 1_I$ denote the $|S|$-dimensional indicator vector of an independent set $I$, where $\vec{(1_I)}_{e'} = 1$ if $e' \in I$ and $0$ otherwise. Since $\vec x$ lies in the matroid polytope, there exist coefficients $p_I \in [0,1]$ for all $I \in \mathcal{I}$ such that
\begin{align*}
    \vec x = \sum_{I : e^* \in I} p_I \vec 1_I + \sum_{I' : e^* \notin I'} p_{I'} \vec 1_{I'}
\end{align*}
and $\sum_{I \in \mathcal{I}} p_I = 1$. Note that $\{e^*\}$ is independent. 
We create a mapping $f:\mathcal I \to \mathcal I$ as follows. For any $I'\in\mathcal{I}$ with $e^*\notin I'$, if $|I'|\ge 2$, the matroid exchange property guarantees an element $e'\in I'$ such that $I''=(I'\setminus\{e'\})\cup\{e^*\}$ is independent~\cite{oxley2011matroids}; set $f(I') = (I'')$. If $|I'|<2$, set $f(I')=\{e^*\}$. Using this mapping, we construct a solution $\vec{y}$ from $\vec{x}$ such that $y_{e^*}=1$.
\begin{align*}
    &\vec y = \sum_{I : e^* \in I} p_I  \vec 1_I + \sum_{I' : e^* \notin I'} p_{I'} \vec 1_{f(I')}, \\
    &\forall e \in (S_g \cup S_{g'}) \setminus \{e^*\}:y_{e} = 0
\end{align*}

By construction, $y_{e^*}=1$. Since $e^*=(i,\{g,g'\})$, the variable $y_{e^*}$ appears only in the following two matching constraints:
\begin{align*}
    \sum_{e\in S_g} y_e &\le 1, \text{ and } \sum_{e\in S_{g'}} y_e \le 1.
\end{align*}
Setting $y_e=0$ for all
$e \in (S_g \cup S_{g'}) \setminus \{e^*\}$
ensures that $\vec y$ is feasible in $\LP(M)$. The mapping $f$ ensures that deleting the coordinate corresponding to $e^*$ from $\vec y$ yields a feasible solution to $\LP(M/e^*)$.

Let $\OPT$ denote the objective value of $\vec x$ in $\mathrm{LP}(M)$. The objective value of $\vec y$ in $\LP(M)$ is at most that of $\vec x$. We now show that the loss in objective value is at most $1$, i.e., $\vec y$ has objective value of at least $\OPT-1$. Observe the following,
\begin{align*}
    \sum_{e\in S_g\setminus\{e^*\}} x_e
    &\le 1-x_{e^*}, \text{ and } \sum_{e\in S_{g'}\setminus\{e^*\}} x_e
    \le 1-x_{e^*}.
\end{align*}
The total utility decrease is at most $2(1-x_{e^*})\le 1$ as $x_{e^*}\ge\frac12$.

Recall that deleting the coordinate corresponding to $e^*$ from $\vec y$ gives a feasible solution for $\LP(M/e^*)$. Since $y_{e^*}=1$, $\LP(M/e^*)$ has an objective value of at least $\OPT - 2$.

The deletion step of \Cref{alg:iterative-matroid-basis} does not reduce the objective value. Whenever \Cref{alg:iterative-matroid-basis} performs a matroid contraction, the objective value of the resulting LP decreases by at most $2$, while it also adds the corresponding item to the solution set increasing the value of the solution set by $1$. Therefore, \Cref{alg:iterative-matroid-basis} computes a $(1/2)$-approximate integral solution to $\mathrm{LP}(M)$.

It only remains to argue that one can find a solution to $\LP(M)$ in polynomial time. It is well known that, if linear optimization over a nonempty rational polytope can be solved in polynomial time, then an optimal extreme point solution can also be computed in polynomial time \cite{schrijver2003combopt}. 
The number of constraints in $\LP(M)$ is exponential; however, given a \emph{separation oracle}, the LP can be solved in polynomial time \cite{schrijver2003combopt}.
A separation oracle decides whether a given candidate solution $\vec x$ satisfies all feasibility constraints, and outputs a violating constraint otherwise. Checking whether any of the (polynomially many) constraints in \eqref{matching-cons} and \eqref{boundary-cons} are violated can be done in polynomial time. Since $r$ is a matroid rank function, checking whether a matroid rank constraint is violated can be done in polynomial time (see \cite[Ch. 5]{lau2011iterative}). 
\end{proof}

    \subsection{Approximating the Nash Welfare for \texorpdfstring{$\SPOC_k$}{SPOC-k} Valuations}
    \Cref{corr:nash-hard} implies that computing a constant-factor approximation to $\maxNSW$ is NP-hard. However, a polynomial-time approximation (\Cref{alg:nash-welfare-local-search}) is possible if we allow both a constant multiplicative factor and a constant additive term. An allocation $X$ offers a $(\alpha,\beta)$ approximation to the max Nash welfare if 
    $$\alpha\left(\prod_{i\in N} (v_i(X_i)+\beta)\right)^{\frac1n} \ge \maxNSW.$$

    In \Cref{alg:nash-welfare-local-search}, we run the \textbf{while} loop as long as some item was assigned to a new agent (controlled by the \texttt{flag} variable). 
    For each agent $i \in N$, we identify a set $H_i(X)$ of agents whose utility is (much) higher than agent $i$'s; the items owned by agents in $H_i(X)$ can be given to $i$. 
    The set of transferable items $G_i(X)$ includes agent $i$'s bundle $X_i$, unassigned items, and items owned by agents in $H_i(X)$. If the set of transferable items $G_i(X)$ contains a complement $\{g_1,\ldots, g_k\} \subseteq G_i(X)$ for agent $i$ \emph{that increases agent $i$'s utility}, we assign the complement to agent $i$. 
    This process repeats until no complement transfers remain. 
    Let us now establish that \Cref{alg:nash-welfare-local-search} approximates the Nash welfare. 
    \begin{algorithm}[htbp]
\caption{Max Nash Welfare via Local Search~($\alpha,\beta$)}
\label{alg:nash-welfare-local-search}
\begin{algorithmic}[1]
\STATE Initialize an allocation $X \gets (X_0;X_1,\dots,X_n)$ 
\STATE $X_0 \gets G$
\FOR{$i \in N$}
\STATE $X_i \gets \emptyset$
\ENDFOR
\STATE $\texttt{flag} \gets1$ 
\WHILE{\texttt{flag}}
\STATE $\texttt{flag} \gets 0$
\FOR{$i \in N$}
    \STATE $H_i(X) \gets \{j \in N: v_j(X_j) > \frac{\alpha -1}{k}(v_i(X_i)+\beta)\}$
    \STATE $G_i(X) \gets \bigcup_{j\in H_i(X)}X_j \cup X_0 \cup X_i$
    \IF{$\exists \{g_1,\ldots, g_k\} \subseteq G_i(X):\Delta_i(X_i,\{g_1,\ldots, g_k\}) = 1$}
        \STATE $X_i \gets X_i \cup \{g_1,\ldots, g_k\}$\label{step:improvement-nash-draft}
        \STATE Remove $g_1,\ldots, g_k$ from their previously assigned bundles.
        \STATE clean the allocation $X$.
        \STATE $\texttt{flag} \gets 1$
    \ENDIF
\ENDFOR
\ENDWHILE
\RETURN $X$ 
\end{algorithmic}
\end{algorithm}

\begin{theorem}\label{thm:5-appx}
    If \Cref{alg:nash-welfare-local-search} terminates when agents have binary $\SPOC_k$ preferences, then it outputs an $(\alpha,\beta)$-approximation to the \maxNSW.
\end{theorem}
\begin{proof}
    Suppose \Cref{alg:nash-welfare-local-search} outputs the allocation $X$. We assume with no loss of generality that agents have purely complementary preferences over $k$-items, i.e., no set of $<k$ items offers them a positive marginal gain. 
    We partition the agents according to utility under $X$ as follows: for a given utility level $j \in \mathbb Z_{\ge0}$, let $U_j(X) = \left\{ i \in N \mid v_i(X_i) = j \right\}$. 
    Let $X^*$ be a clean Nash welfare maximizing allocation. 
    There is some agent $i \in N$ such that $v_i(X_i^*)>v_i(X_i)$; otherwise $\NSW(X)\ge \NSW(X^*)$ and we are done.

Since both $X$ and $X^*$ are clean allocations and $v_i$ is a MRF on $k$-sized complements, there are at least $v_i(X_i^*) - v_i(X_i)$ complements in $X_i^*$ that offer agent $i$ a marginal utility of $1$ given that they have $X_i$.   
Consider a complement $\{g_1,\ldots, g_k\} \subseteq X_i^*$ such that 
$$\Delta_i(X_i,\{g_1,\ldots, g_k\})= 1.$$ 
Let $G_i(X)$ and $H_i(X)$ be defined as in \Cref{alg:nash-welfare-local-search}, and let $G_i^c(X) = G\setminus G_i(X)$ be all items that contain no  complements that offer agent $i$ a marginal utility of $1$ under \Cref{alg:nash-welfare-local-search}.
Since \Cref{alg:nash-welfare-local-search} terminated without assigning $\{g_1,\dots, g_k\}$ to agent $i$, it must be that one of the items from $\{g_1,\ldots,g_k\}$ is in $G_i^c$.

Let $S_i$ be the union of all complementary bundles in $X_i^*$ that can improve agent $i$'s welfare given $X_i$; then there are at least $v_i(X_i^*) -v_i(X_i)$ items in $S_i \cap G_i^c(X)$. 

\begin{restatable}{lemma}
{lemnswalgebra}\label{lem:NSWalgebra}
    Let $\NSW_\beta(X)$ denote $\left(\prod_{i \in N}(v_i(X_i)+\beta)\right)^{\frac1n}$. Then:
    \[\log \frac{\NSW(X^*)}{\NSW_\beta(X)}\le \log\left(1 + \frac{1}{n}\sum_{j=0}^{\floor{\frac mk}}\frac{1}{j+\beta}\sum_{i\in U_j(X)}|S_i\cap G_i^c(X)|\right)\]
\end{restatable}
The proof of \Cref{lem:NSWalgebra} is in the appendix.
We upper-bound the right-hand side of the inequality in \Cref{lem:NSWalgebra} via a factor-revealing LP. The following lemma establishes the desired bound.
\begin{restatable}{lemma}{lemuplogsum}\label{clm:up-log-sum-draft}
    $\frac{1}{n}\sum_{j=0}^{m+1}\frac{1}{j+\beta}\sum_{i\in U_j(X)}|S_i\cap G_i^c(X)| \le \alpha - 1$.
\end{restatable}
\begin{proof}
Fix any agent $i \in U_j(X)$. Note that $G_i^c(X)$ contains all items allocated to agents whose utility is $\le \frac{\alpha-1}{k}(j+\beta)$. Define
\[
L_j(X) := \bigcup_{j' = 0}^{\min\{\frac{\alpha-1}{k}(j+\beta),\frac mk\}} \bigcup_{q\in U_{j'}(X)} X_q,
\]
i.e., $L_j(X)$ is the set of all items assigned in $X$ to agents whose utility is at most $\min\{\frac{\alpha-1}{k}(j+\beta), \floor{\frac mk}\}$. We take the minimum of $\frac{\alpha-1}{k}(j+\beta)$ and $\floor{\frac mk}$ since the maximal utility of any agent with $\SPOC_k$ preferences over $m$ items is $\floor{\frac mk}$.
Observe that the sets are nested, i.e., $L_0(X)\subseteq L_1(X) \subseteq \cdots \subseteq L_{\floor{\frac mk}}(X)$.
Furthermore,
\[
|L_j(X)| = \sum_{j'=0}^{\min\{\frac{\alpha-1}{k}(j+\beta),\floor{\frac mk}\}} kj'|U_{j'}(X)|.
\]

Consider the linear program:
\begin{align}
\text{maximize} \quad & \sum_{j=0}^{\floor{\frac mk}} \frac{x_j}{j+\beta} \label{eq:nash-lp}\\
\text{subject to} \quad 
& \sum_{j=0}^{t} x_j \le |L_t(X)|, \quad t = 0, \ldots, \floor{\frac mk},\notag \\
& x_j \ge 0, \quad j = 0, \ldots, \floor{\frac mk}.\notag
\end{align}

Establishing the two claims below (proofs in the appendix) regarding LP \eqref{eq:nash-lp} concludes the proof. 

\begin{restatable}{claim}{clmlpfeasibleone}\label{subclm:feasibility-draft}
    The choice $x_j = \sum_{i\in U_j(X)}|S_i\cap G_i^c(X)|$ is feasible for the LP.
\end{restatable}

\begin{restatable}{claim}
{clmlpfeasible}\label{subclm:optimality-draft}
    The solution $x_j^* = |L_j(X)\setminus L_{j-1}(X)|$, with $L_{-1}(X)=\emptyset$, is optimal for the LP.
\end{restatable}

By Claim \ref{subclm:feasibility-draft} and \ref{subclm:optimality-draft}, the LP objective is attained at $\vec x^*$ such that $x_j^* = |L_j(X)\setminus L_{j-1}(X)|$, and this value upper bounds 
$\sum_{j=0}^{\floor{\frac mk}}\sum_{i\in U_j(X)}$ $\frac{1}{j+\beta}|S_i\cap G_i^c(X)|$. Hence,
\begin{align}
    \sum_{j=0}^{\floor{\frac mk}}\frac{x_j^*}{j+\beta}
    &= \sum_{j=0}^{\floor{\frac mk}}\frac{1}{j+\beta}|L_j(X)\setminus L_{j-1}(X)|\label{eq:nash-upperbound-3-draft}
\end{align}
Note that $|L_j(X) \setminus L_{j-1}(X)|=|L_j(X)| - |L_{j-1}(X)|$. Therefore, we get that $|L_j(X) \setminus L_{j-1}(X)|$ equals:
\begin{align*}
 &\sum_{j' = 0}^{\min\{\frac{\alpha-1}{k}(j+\beta),\floor{\frac mk}\}}kj'|U_{j'}(X)|-
        \sum_{j' = 0}^{\min\{\frac{\alpha-1}{k}(j+\beta-1),\floor{\frac mk}\}}kj'|U_{j'}(X)|\\
=        & k\left(\sum_{j' = \min\{\frac{\alpha-1}{k}(j+\beta-1)+1,\floor{\frac mk}\}}^{\min\{\frac{\alpha-1}{k}(j+\beta),\floor{\frac mk}\}}j'|U_{j'}(X)|\right)
\end{align*}

    We observe that $\frac{\frac{\alpha-1}{k}(j+\beta-1)+1}{j+\beta} \le\ldots\le \frac{\frac{\alpha-1}{k}(j+\beta)}{j+\beta}\le \frac{\alpha-1}{k}$.
    
    Thus, \Cref{eq:nash-upperbound-3-draft} equals
    \begin{align*}
    &\sum_{j=0}^{\floor{\frac mk}}\frac{1}{j+\beta}
       \sum_{j' = \min\{\frac{\alpha-1}{k}(j+\beta-1)+1,\floor{\frac mk}\}}^{\min\{\frac{\alpha-1}{k}(j+\beta),\floor{\frac mk}\}} kj'|U_{j'}(X)| \\
       \le& 
    \sum_{j=1}^{m+1} (\alpha-1)\sum_{j' = \min\{\frac{\alpha-1}{k}(j+\beta-1)+1,\floor{\frac mk}\}}^{\min\{\frac{\alpha-1}{k}(j+\beta),\floor{\frac mk}\}} |U_{j'}(X)|\le (\alpha -1)n. 
\end{align*}
The last inequality is due to the fact that the sets $U_{j'}(X)$ partition the agent set $N$, and that each set appears exactly once in the summation.
\end{proof}

The bound shown in \Cref{clm:up-log-sum-draft} implies \Cref{thm:5-appx}.
\end{proof}
Now we derive a relation between $\alpha,\beta,$ and $k$ for which \Cref{alg:nash-welfare-local-search} terminates in polynomial time.
\begin{restatable}{lemma}{corNashWelfarePoly}\label{lem:Nash-Welfare-Poly}
    For binary $\SPOC_k$ agents, if $\alpha - 1\ge k^2$, $\beta(\alpha-1)\ge k^2$, \Cref{alg:nash-welfare-local-search} terminates in time polynomial in $m$ and $n$. 
\end{restatable}
For example, \Cref{lem:Nash-Welfare-Poly} and \Cref{thm:5-appx} imply that we can efficiently compute a $(5,1)$, a $(7,\frac23)$ or a $(9,\frac12)$ approximation to the \maxNSW for binary $\SPOC_2$ agents. The proof of \Cref{lem:Nash-Welfare-Poly} is in the appendix.

\section{Computing Fair and Approximately Efficient Allocations for SPOC agents}\label{sec:usw-efx}
Let us now turn to the problem of efficiently computing fair and approximately optimal allocations. Throughout this section we assume that agents have \emph{purely complementary} binary $\SPOC_k$ preferences, i.e., each synergistic bundle is of size exactly $k$.
In this setting, \Cref{alg:clean-efx} computes a $\frac{1}{k+1}$-\maxUSW and EF-X allocation where the utility of each agent $i\in N$ is at least $\frac{\MMS_i - k}{k+1+\frac{1}{n}}$. 

\begin{algorithm}[htbp]
\caption{Partial EF-X allocation}
\label{alg:clean-efx}
\begin{algorithmic}[1]
    \STATE $H^0 \gets G$
    \STATE $j \gets 0$
    \STATE Initialize $A^0=(A_1^0,\ldots,A_n^0)$ with
           $A_i^0=\varnothing$ for every $i\in N$

    \WHILE{$\exists i\in N: v_i(A_i^j)<v_i(H^j)$}
        \STATE $T^j \gets
               \{i\in N:v_i(A_i^j)<v_i(H^j)\}$
        \STATE Choose
               $i^*\in\arg\min_{i\in T^j}v_i(A_i^j)$,
               breaking ties arbitrarily \label{state:select-i}
        \STATE Choose an arbitrary clean subset
               $X^j\subseteq H^j$ for agent $i^*$ such that
$v_{i^*}(X^j)=v_{i^*}(A_{i^*}^j)+1$
        \STATE $A_{i^*}^{j+1}\gets X^j$
        \STATE $A_i^{j+1}\gets A_i^j$ for every $i\ne i^*$
        \STATE $H^{j+1}\gets
               (H^j\setminus X^j)\cup A_{i^*}^j$
        \STATE $j\gets j+1$
    \ENDWHILE
    \RETURN $A^j$
\end{algorithmic}
\end{algorithm}
\Cref{alg:clean-efx} returns a partial allocation. We first show our guarantees on this partial allocation. Then, we complete the partial allocation without violating any of those guarantees.
\begin{lemma}\label{lm:cleanEFX}
    \Cref{alg:clean-efx} returns a clean EF-X allocation.
\end{lemma}
\begin{proof}
Our proof is by induction on the number of iterations of the \textbf{while} loop in \Cref{alg:clean-efx}. 
The initial empty allocation $A^0$ trivially satisfies both properties.

Suppose that after $t$ iterations of the \textbf{while} loop,
$A^t=(A_1^t,\ldots,A_n^t)$ is clean and EF-X. 
Let $i^*$ be the agent selected in line \ref{state:select-i} of \Cref{alg:clean-efx} at time $t+1$.
By definition, agent $i^*$ receives a clean bundle, and $v_{i^*}(A_{i^*}^{t+1})=v_{i^*}(A_{i^*}^t)+1$.
For every agent $i\ne i^*$, we have $A_i^{t+1}=A_i^t$. Thus all bundles remain clean. It remains to prove that the allocation is EF-X. Since the only bundle that has changed is agent $i^*$'s, and agent $i^*$'s utility has increased, we need only verify that no agent $i \ne i^*$ envies agent $i^*$'s new bundle.
\begin{enumerate}[label = \textbf{Case \arabic*:}, wide]
\item suppose that $v_i(A_i^{t+1}) \ge v_{i^*}(A_{i^*}^{t})$. 
Since the bundle $A_{i^*}^t$ is clean, $v_{i^*}(A_{i^*}^t) = \frac{|A_{i^*}^t|}{k}$. 
That is, $v_i(A_i^t) \ge v_i(A_{i^*}^t)$. Since agent $i$ is a binary $\SPOC_k$ agent, $v_i(A_{i^*}^t) \le \frac{|A_{i^*}^t|}{k}$. Since agent $i^*$'s utility increased by $1$, we must have $|A_{i^*}^{t+1}| = |A_{i^*}^t|+k$.   
After removing any item $g$ from $A_{i^*}^{t+1}$ the maximum utility that agent $i$ can extract from $A_{i^*}^{t+1}\setminus \{g\}$ is at most $\left\lfloor\frac{|A_{i^*}^t|+k-1}{k}\right\rfloor = v_i(A_{i^*}^t)\le v_i(A_i^t) = v_i(A_i^{t+1})$, i.e., agent $i$ does not envy agent $i^*$ after removing any item at time $t+1$.

\item Suppose that $v_i(A_i^{t+1}) = v_i(A_i^t)<v_{i^*}(A_{i^*}^t)$. If agent $i$ belongs to $T^t$ they should be selected at Line \ref{state:select-i}, since $v_i(A_i^t) < v_{i^*}(A_{i^*}^t)$; thus, $i\notin T^t$.
Since $A_{i^*}^{t+1} \subseteq H^t$, we have $v_i(A_{i^*}^{t+1}) \le v_i(H^t)\le v_i(A_i^t)$, i.e., agent $i$ does not envy agent $i^*$.
\end{enumerate}
In either case, the allocation $A^{t+1}$ remains EF-X. 
\end{proof}

\begin{restatable}{lemma}{lemCleanEFXRuntime}\label{lem:clean-efx-runtime}
    \Cref{alg:clean-efx} terminates after at most
    $\left\lfloor\frac{|G|}{k}\right\rfloor$ iterations.
\end{restatable}
\noindent The proof of \Cref{lem:clean-efx-runtime} is in the appendix.

\begin{lemma}\label{lm:usw-guarantee}
    Let $A$ be the allocation output by \Cref{alg:clean-efx}. 
    Then
    \[
        \sum_{i=1}^n v_i(A_i)\ge\frac{1}{k+1}\maxUSW.
    \]
\end{lemma}
\begin{proof}
Suppose that \Cref{alg:clean-efx} terminates after $b$ rounds.
Let $A^*=(A_1^*,\ldots,A_n^*)$ be a clean $\maxUSW$ allocation. 
For each agent $i\in N$, let $L_i^*$ be the union of complements of $i$ in $A_i^*$ that share at least one good with $\cup_j A_j = G\setminus H^b$. 
Now observe that all the complements of $i$ in $A_i^*$ not contained in $L_i^*$ must be in $H^b$. Hence, $A_i^*\setminus L_i^*\subseteq H^b$. 
Therefore, $\sum_{i\in N} v_i(A_i^*)$ is bounded by:
\begin{align}
     \sum_{i\in N} \left(v_i(L_i^*) + v_i(A_i^*\setminus L_i^*)\right)\le\sum_{i\in N} v_i(L_i^*) + \sum_{i\in N}v_i(H^b)\label{eq:ub-efx}.
\end{align}
Furthermore:
\begin{align}
    \sum_{i\in N}v_i(L_i^*)&=\frac{|\cup_{i\in N}L_i^*|}{k}\le k \cdot\frac{|\cup_{i\in N}A_i|}{k}= k \sum_{i\in N}v_i(A_i^b)\label{eq:pool}
\end{align}

Now, by the termination condition of the \textbf{while} loop we have $\sum_{i\in N}v_i(H^b)\le \sum_{i\in N}v_i(A_i^b)$. By \Cref{eq:ub-efx,eq:pool}:
\begin{align*}
    \sum_{i\in N} v_i(A_i^*)=\maxUSW\le (k+1) \sum_{i\in N}v_i(A_i^b).
\end{align*}
\end{proof}
In particular, \Cref{lm:usw-guarantee} implies the existence of a $\frac13$-\maxUSW+EF-X allocation for binary $\SPOC_2$ preferences.   
\begin{lemma}\label{lm:mms-guarantee}
    Let $A$ be the allocation output by \Cref{alg:clean-efx}. Then for every agent $i\in N$,
    $
    v_i(A_i)\ge \frac{\MMS_i - k}{k+1+\frac{1}{n}}.
    $
\end{lemma}
\begin{proof}
    Suppose that \Cref{alg:clean-efx} terminates after $b$ iterations.
    Fix an agent $i$, and set $s=v_i(A_i)$. 
    Let $t$ be the first iteration of the \textbf{while} loop after which some agent's utility increases from $s+1$ to $s+2$. 
    If such an agent does not exist then we set $t = b$. 
    Therefore, at $t$-th iteration of the \textbf{while} loop, all agents have a utility of at most $s+1$. Also, note that $v_i(A_i^t)\le s$. 
    Therefore, $v_i(H^t)\le v_i(A_i^t)\le s$ (otherwise, agent $i$ would have been chosen by \Cref{alg:clean-efx}).

    Now consider a clean (with respect to $i$'s valuation) MMS partition $X = (X_0;X_1,\ldots,X_n)$ of agent $i$; w.l.o.g. we assume that $v_i(X_1) = v_i(X_2)=\ldots=v_i(X_n)=\MMS_i$. 
    Let $\{g_1,\ldots, g_k\}$ be a complementary bundle for $i$ in $G\setminus X_0$. 
    There are three possibilities. 
    \begin{enumerate*}[label = {(\arabic*)}]
        \item $\{g_1,\ldots,g_k\}\subseteq A_i^t$ \label{mms-case-subset}
        \item $\{g_1,\ldots,g_k\}\cap \cup_{j\in N\setminus\{i\}}A_j^t\ne \emptyset$ or \label{mms-case-others}
        \item $\{g_1,\ldots,g_k\}$ is entirely contained in $H^t$ \label{mms-case-useless}
    \end{enumerate*}
    There can be at most $s$ complements of $i$ of type \ref{mms-case-subset}, at most $k(s+1)(n-1)$ many complements of type \ref{mms-case-others}, and at most $ns$ many complements of type \ref{mms-case-useless}: as $v_i(H^t)\le s$, every bundle $X_j$ can contain at most $s$ independent complements of $H^t$.
    Thus:
    \begin{align*}
        n\MMS_i &\le s + k (s+1) (n-1) + ns\le ((k+1)n + 1) s + kn
    \end{align*}
    Therefore, $v_i(A_i) = s \ge \frac{\MMS_i - k}{k+1+\frac{1}{n}}$.
\end{proof}

\begin{restatable}{theorem}{thmCompleteEFX}\label{thm:completeEFX}
Given a partial EF-X allocation $A^0=(A_1^0,\ldots,A_n^0)$,
\Cref{alg:complete-efx} returns a complete EF-X allocation $A$ in polynomial time such that
$v_i(A_i)\ge v_i(A_i^0)$ for every $i\in N$.
\end{restatable}
The proof of \Cref{thm:completeEFX} is in the appendix.

Completing the allocation output by \Cref{alg:clean-efx} using \Cref{alg:complete-efx} yields:
\begin{theorem}\label{thm:efx-usw-mms}
    Under binary $\SPOC_k$ valuations, we can efficiently compute a complete $\frac{1}{k+1}$-\maxUSW EF-X allocation $A$ where for each agent $i$, $v_i(A_i) \ge \frac{\MMS_i-k}{k+1+\frac1n}$.
\end{theorem}

Finally, we bound the welfare achievable by approximately EF allocations. 
An allocation $X=(X_1,\ldots,X_n)$ is envy free up to $\ell$ items (EF-$\ell$) if for any two agents $i\neq j$, there exists a set of $\le \ell$ items $S\subseteq X_j$ such that $v_i(X_i)\ge v_i(X_j\setminus S)$. 
Now we show the following.
\begin{theorem}
For every $\epsilon\in(0,1)$ and every integer $\ell\ge1$,
there exists an instance where agents have binary $\SPOC_2$
valuations such that every EF-$\ell$ allocation $X$ satisfies
$\USW(X)\le(\frac12+\epsilon)\maxUSW$.
\end{theorem}
\begin{proof}
We construct an instance with $n\ge2$ agents and $m=4t$ items
$\{a_j,b_j,a_j',b_j':j\in[t]\}$.
The set of complements for agent 1 is
$L_1=\{\{a_j,a_j'\},\{b_j,b_j'\}:j\in[t]\}$.
For each of the remaining agents, the set of complements is
$\{\{a_j,b_j\}:j\in[t]\}$.
We assume that all agents have binary additive valuations
over their complements.
We can achieve the maximal possible \USW ($2t$) by assigning all items to agent $1$.
Suppose $X=(X_1,\ldots,X_n)$ is an EF-$\ell$ allocation.

Let $r$ be the number of $\{a_j,b_j\}$ pairs 
with $j\in[t]$ contained in $X_1$.
Therefore, for every agent $i \ne 1$, we have $v_i(X_1)=r$.
Since $X$ is an EF-$\ell$ allocation and removing at most
$\ell$ items destroys at most $\ell$ such pairs, we have
$v_i(X_i)\ge r-\ell$ for every $i\ne 1$.
The pairs counted by $r$ and those contributing to the
utilities of the remaining agents are distinct. 
Hence,
\begin{align*}
r+\sum_{i=2}^{n}v_i(X_i)\le m\Rightarrow r+(n-1)(r-\ell)\le m \Rightarrow r\le\frac{m+(n-1)\ell}{n}.
\end{align*}

Each four-item block
$\{a_j,b_j,a_j',b_j'\}$ contributes at most $2$ to
$\USW(X)$ if both $a_j$ and $b_j$ belong to $X_1$; otherwise, it contributes at most $1$.
Therefore, $\USW(X)\le 2r+(\frac m4-r)=r+\frac m4$.

Using the preceding bounds, we get
\begin{align*}
\frac{\USW(X)}{\maxUSW}\le\frac{r+\frac m4}{\frac m2}\le\frac{\frac{m+(n-1)\ell}{n}+\frac m4}{\frac m2}=\frac12+\frac2n+\frac{2(n-1)\ell}{mn}.
\end{align*}
For any fixed $\ell$, we can choose sufficiently large $t$
and $n$, so that this ratio
is at most $1/2+\epsilon$, concluding the proof.
\end{proof}

\section{Conclusions and Future Work}\label{sec:conclusions}
In this work, we examine the problem of finding fair and efficient allocations when agents have $\SPOC_k$ preferences. 
We present both hardness results and combinatorial methods for computing approximate \maxUSW and \maxNSW allocations, as well as methods for computing fair, approximately efficient allocations. Our results highlight the technical challenges introduced by complementary preferences. While the results we present offer a fairly comprehensive coverage of the binary $\SPOC_k$ setting, it is unclear how one would extend our results beyond binary marginal gains. While $\SPOC_k$ instances with general marginal gains inherit known intractability results for submodular agents, restricting agents' marginal gains to a fixed set may yield positive results, as it does for submodular preferences with restricted marginals \cite{akrami2022mnw,amanatidis2021mnw,cousins2023bivalued,ebadian2022bivaluedchores,garg2021bivaluedadditive,garg2022bivaluedchores}
. 

\paragraph{Acknowledgments.}
Viswanathan and Zick are supported by NSF Grants (CAREER \#2441296, RI:Small \#2337382).



\bibliographystyle{alpha} 
\bibliography{ref.bib}


\appendix
    \onecolumn
    \section{Missing Proofs from \texorpdfstring{\Cref{sec:complexity}}{Section 3}}
Proof of \Cref{lm:lp-1}.
\lemLP*
    \begin{proof}
Given an allocation $X=(X_1,\ldots,X_n)$, we assume that agents' bundles consist only of complementary item pairs: any singleton items do not offer agents utility and are assumed to be unassigned. We define a vector $\vec x$ as follows: given an agent $i$ and complementary item pair $\{g,g'\}\subseteq X_i$ assigned to agent $i$, we set $x_e = 1$ where $e=(i,\{g,g'\})$. 
For every subset of edges $T \subseteq S$ and every agent $i \in N$, let $T_i$ be the edges in $T$ of the form $(i,\{g,g'\})$. We observe that $T_1,\dots,T_n$ partition $T$. 
Consider the set of items $R(T_i) = \{g,g':(i,\{g,g'\})\in T_i\}$. For any edge $e =(i,\{g,g'\})\in T_i$, $x_e=1$ iff the pair $g,g'$ is in $X_i$. Thus, $\sum_{e \in T_i} x_e = |R(T_i) \cap X_i|/2$. Since $v_i$ is binary submodular and $v_i(X_i) = |X_i|/2$, for any $Q \subseteq X_i$ we have $v_i(Q) = |Q|/2$. In particular, $|R(T_i) \cap X_i|/2 = v_i(R(T_i) \cap X_i)$. Since $v_i$ is monotone, $v_i(R(T_i)\cap X_i) \le v_i(R(T_i))$. Putting this all together we get that $r(T)$ equals
\begin{align*}
    \sum_{e \in T}\sum_{i= 1}^nv_i(\{g,g':(i,\{g,g'\}) \in T\}) =& \sum_{i = 1}^n v_i(R(T_i))\ge \sum_{i=1}^nv_i(R(T_i) \cap X_i)\\=&\sum_{i =1}^n \sum_{e \in T_i}x_e = \sum_{e \in T}x_e
\end{align*}
 Thus, \eqref{rank-cons} holds. Since $X$ is a valid allocation, for any item $g \in G$ we have $\sum_{e \in S_g} x_e \le 1$, and by construction $x_e \in \{0,1\}$. Thus, the solution $\vec x$ constructed from the allocation $X$ satisfies all of the constraints of $\LP(M)$, and the value of $\vec x$ is trivially $\sum_{i = 1}^n v_i(X_i)$. 

Conversely, let $\vec x'$ be an integral feasible solution to the LP. We construct an allocation $X'=(X'_1,\ldots,X'_n)$ as follows: for every agent $i \in N$, let $X'_i$ be the set of all item pairs $\{g,g'\}$ for which the edge $e = (i,\{g,g'\})$ satisfies $x'_e=1$. Constraint~\ref{matching-cons} ensures that each good $g \in G$ is assigned to at most one agent, so $X'$ is a valid partition of $G$. Finally, by construction, the objective value of the LP equals the $\USW$ of the allocation $X'$. 
\end{proof}

Proof of \Cref{lm:extreme-point}.
\lemmaextremepoint*
\begin{proof}
   Suppose for contradiction that for every element $e \in S$, $x_e \in \left(0, \tfrac{1}{2}\right)$. 
   Let $A$, $B$, and $C$ denote the sets of the tight constraints corresponding to \eqref{rank-cons}, \eqref{matching-cons}, and \eqref{boundary-cons}, respectively. 
   Since $\vec x$ is an extreme point solution, these tight constraints are linearly independent and satisfy $|A| + |B| + |C| = |S|$.
Moreover, because $x_e \in (0, \tfrac{1}{2})$ for all $e \in S$, none of the boundary constraints in \eqref{boundary-cons} are tight, implying $|C| = 0$. Consequently, $|A| + |B| = |S|$.
\begin{enumerate}[label=\textbf{Case \Roman*:},align=left,wide]
\item Suppose that $|A| > \tfrac{|S|}{3}$. The tight rank constraints in $A$ form a \emph{chain}; that is, there exists a sequence of subsets of elements of $S$, $\mathcal{L} = \{\chi_0, \chi_1, \ldots, \chi_{|A|}\}$ such that $\emptyset = \chi_0 \subset \chi_1 \subset \cdots \subset \chi_{|A|}$, where $\chi_j$ satisfies $\sum_{e \in \chi_j}x_e = r(\chi_j)$ \cite{schrijver2003combopt}.

Suppose, for the sake of contradiction, that for every $t \in [|A|-1]$ we have
$|\chi_{t+1} \setminus \chi_t| \ge 3$.
Then
\[
|\chi_{|A|}| = \sum_{t=0}^{|A|-1} |\chi_{t+1} \setminus \chi_t|
\ge 3|A| > |S|,
\]
which contradicts the fact that $\chi_{|A|} \subseteq S$. 
Hence, there exists some $t \in [|A|-1]$ such that $|\chi_{t+1} \setminus \chi_t| \le 2$.
Since $r(\chi_{t+1})\ge r(\chi_{t})+1$ 
and the constraints are tight for $\chi_t$ and $\chi_{t+1}$, this implies that
$\sum_{e \in \chi_{t+1}}x_e \ge \sum_{e \in \chi_t}x_e + 1$. Since the number of summands in the left-hand side is at most 2 more than the number of summands in the right-hand side, 
there exists an element $e \in \chi_{t+1} \setminus \chi_t$ with $x_e \ge \tfrac{1}{2}$, a contradiction.
\item Now suppose that $|A| < \tfrac{|S|}{3}$, which implies that $|B| > \tfrac{2|S|}{3}$. 
Assume, for the sake of contradiction, that for every good $g$ such that $S_g \in B$, we have $|S_g| \ge 3$. 
Then,
\[
\sum_{g} |S_g| \ge\sum_{g : S_g \in B} |S_g| \ge 3|B| > 2|S|.
\]
For any element $e = (i,\{g,g'\}) \in S$, the element $e$ appears in exactly two sets, namely $S_g$ and $S_{g'}$. Consequently, $\sum_{g} |S_g| \le 2|S|$,
which contradicts the inequality above. Hence, there exists a good $g$ such that $S_g \in B$ and $|S_g| \le 2$. It follows that there exists an element $e \in S_g$ with $x_e \ge \tfrac{1}{2}$, yielding the desired contradiction.
\item 
Finally, consider the case where $|A| = \tfrac{|S|}{3}$ and $|B| = \tfrac{2|S|}{3}$. 
Let $\mathcal{L} = \{\chi_0, \chi_1, \ldots, \chi_{|A|}\}$
denote the list of tight rank constraints in $A$, as in Case I. 
If there exists some $t \in [|A|-1]$ such that $|\chi_{t+1} \setminus \chi_t| \le 2$, then the desired conclusion follows immediately from the argument in Case~I.
We assume that $|\chi_{t+1} \setminus \chi_t| \ge 3$ for all $t \in [|A|-1]$.
It follows that
\[
|\chi_{|A|}| = \sum_{t=0}^{|A|-1} |\chi_{t+1} \setminus \chi_t|
\ge 3|A| = |S|.
\]
Thus, the constraint of $\chi_{|A|}$ must be
\begin{align}
\sum_{e \in S} x_e = r(S) \label{eq:Aconstraints}
\end{align}
Next, if there exists a good $g \in G$ such that $S_g \in B$ and $|S_g| \le 2$, we are again done. 
Thus, assume that for every $g \in G$ with $S_g \in B$, we have $|S_g| \ge 3$. Then,
\[
\sum_{g : S_g \in B} |S_g| \ge 3|B| = 2|S|.
\]
However, as shown in Case~II, $\sum_{g} |S_g| \le 2|S|$ and equality must hold.
Thus, if we sum the rows corresponding to the constraints in $B$, we get a vector of all 2s.  

The constraint \eqref{eq:Aconstraints} is a row of 1s in the constraint matrix, whereas the sum of constraints in $B$ sums to a row of 2s. 
Thus, the tight constraints indexed by $A \cup B$ are linearly dependent, contradicting the fact that we picked a set of linearly independent constraints $A \cup B$.
\end{enumerate}
\end{proof}

\lemnswalgebra*
\begin{proof}
Note that $\log \frac{\NSW(X^*)}{\NSW_\beta(X)}$ is:
\begin{align}
         & \log\left(\frac{\prod_{i \in N}v_i(X_i^*)}{\prod_{i \in N}(v_i(X_i)+\beta)}\right)^{\frac1n} = \frac1n\sum_{i \in N}\log\frac{v_i(X_i^*)}{v_i(X_i)+\beta}\notag\\
        =&\frac1n\sum_{i \in N}\log\frac{v_i(X_i) + (v_i(X_i^*)-v_i(X_i))}{v_i(X_i)+\beta} \le  \frac1n\sum_{i \in N}\log\left(1+\frac{|S_i\cap G_i^c(X)|}{v_i(X_i)+\beta}\right)\notag\\
        \le & \frac1n\sum_{j=0}^{\floor{\frac mk}}\sum_{i\in U_j(X)}\log\left(1+\frac{|S_i\cap G_i^c(X)|}{j+\beta}\right)\label{eq:nash-upperbound-1-draft}
\end{align}
Since the arithmetic mean is at least the geometric mean, we have 
\begin{align*}
\frac{\sum_{i\in U_j(X)} \left(1+\frac{|S_i\cap G_i^c(X)|}{j+\beta}\right)}{|U_j(X)|}
\ge \left(\prod_{i\in U_j(X)}\left(1+\frac{|S_i\cap G_i^c(X)|}{j+\beta}\right)\right)^{\frac{1}{|U_j(X)|}}.
\end{align*}
Thus, \Cref{eq:nash-upperbound-1-draft} is upper-bounded by
\begin{align}
        \frac1n\sum_{j=0}^{m+1}|U_j(X)|\log\left(\frac{|U_j(X)| + \frac{\sum_{i\in U_j(X)}|S_i\cap G_i^c(X)|}{j+\beta}}{|U_j(X)|}\right) \label{eq:nash-upperbound-2-draft}
\end{align}
We apply the log-sum inequality: for nonnegative real numbers $a_1,\ldots,a_n$ and $b_1,\ldots,b_n$,
\[
\sum_{j=1}^{n} a_j \log \frac{b_j}{a_j}\le \left(\sum_{j = 1}^n a_j\right) \log \frac{\sum_{j = 1}^n b_j}{\sum_{j = 1}^n a_j}.
\]
We apply this inequality with $a_j = |U_j(X)|$ and $b_j = |U_j(X)| + \sum_{i\in U_j(X)}\frac{|S_i\cap G_i^c(X)|}{j+\beta}$. We note that $\sum_{j=0}^m |U_j(X)| = n$, since every agent belongs to exactly one of the sets $U_j(X)$. Thus, \Cref{eq:nash-upperbound-2-draft} is upper bounded by
\begin{align*}
        \log\left(\frac{\sum_{j=0}^{\floor{\frac mk}}|U_j(X)| + \sum_{j=0}^{\floor{\frac mk}}\sum_{i\in U_j(X)}\frac{|S_i\cap G_i^c(X)|}{j+\beta}}{n}\right)
        &= \\ 
        \log\left(1 + \frac{1}{n}\sum_{j=0}^{\floor{\frac mk}}\frac{1}{j+\beta}\sum_{i\in U_j(X)}|S_i\cap G_i^c(X)|\right)&
\end{align*}
\end{proof}

\clmlpfeasibleone*
\begin{proof}
    Observe that all goods appearing in the $k$-complements of $S_i$ are contained in $X_i^*$. Since the allocation $X^*$ partitions $G$, for any two agents $i$ and $i'$, the sets $S_i\cap G_i^c(X)$ and $S_{i'}\cap G_{i'}^c(X)$ are disjoint. Therefore, 
$\sum_{j=0}^{t}\sum_{i\in U_j(X)}|S_i\cap G_i^c(X)|$ is bounded by the total number of goods assigned to agents with utility at most $\frac{\alpha-1}{k}(t+\beta)$, which is precisely $|L_t(X)|$.
\end{proof}

\clmlpfeasible*
\begin{proof}
    We first observe that the solution $\vec x^* = x_0^*,\dots,x_{\left\lfloor\frac mk\right\rfloor}^*$ is feasible; since the sets $L_j(X)$ are nested: 
    $$\sum_{j = 1}^t x_j^* = \sum_{j = 1}^t |L_j(X)\setminus L_{j-1}(X)| = |L_t(X)|.$$
    Next, consider any other feasible solution $\vec x = x_1,\dots,x_{m+1}$ to the LP.  
If all constraints are tight, then the solution $\vec x$ must coincide with $\vec{x}^*$. Suppose for contradiction that $\vec x \ne \vec x^*$ is optimal, and let $t_{\min}$ be the smallest index such that
\[
\sum_{j=1}^{t_{\min}} x_j < |L_{t_{\min}}(X)|.
\]
This implies that $x_{t_{\min}} < |L_{t_{\min}}(X)\setminus L_{t_{\min}-1}(X)|$. Let $k > t_{\min}$ be the smallest index for which $x_k > 0$.
\begin{enumerate}[label = \textbf{Case \Roman*:},align =left]
    \item Suppose no such $k$ exists. 
    Then setting $x_{t_{\min}} = |L_{t_{\min}}(X)\setminus L_{t_{\min}-1}(X)|$ 
maintains feasibility of all constraints, while strictly increasing the objective value, a contradiction to the optimality of $\vec x$.
\item Suppose such a $k$ exists. Define 
\[
\delta = \min\left\{\, |L_{t_{\min}}(X)\setminus L_{t_{\min}-1}(X)| - x_{t_{\min}},\; x_k \right\}.
\]
Now increase $x_{t_{\min}}$ by $\delta$ and decrease $x_k$ by $\delta$. This preserves feasibility of all constraints, yet strictly improves the objective value, again yielding a contradiction.
\end{enumerate}
\end{proof}

To prove \Cref{lem:Nash-Welfare-Poly}, we first prove the following lemma.
\begin{restatable}{lemma}{lemNashIncrease}\label{lem:Nash-increase}
    Suppose $\beta(\alpha-1)\ge k^2$ and $\alpha - 1\ge k^2$. Whenever step \ref{step:improvement-nash-draft} occurs in \Cref{alg:nash-welfare-local-search}, one of the following occurs:
    \begin{enumerate}
    \item the number of agents with utility $0$ decreases, or
    \item the Nash welfare increases by a factor of at least $\left(1+\frac1{m^4}\right)^{\frac1n}$.
    \end{enumerate}
\end{restatable}
\begin{proof}
    According to the Nash criterion, an allocation $X$ is better than an allocation $Y$ if the number of agents with utility $0$ under $X$ is smaller than the number of agents with utility $0$ under $Y$. If both allocations have the same number of $0$ utility agents, then $X$ is better than $Y$ iff $\prod_{i:v_i(X_i) > 0} v_i(X_i) > \prod_{i:v_i(Y_i) > 0}v_i(Y_i)$. 

    If agent $i$ only steals unassigned items in Step \ref{step:improvement-nash-draft}, then agent $i$'s utility improves without harming other agents, and the claim trivially holds. Thus, we assume that at Step \ref{step:improvement-nash-draft}, at least one agent $j$ had items stolen from them. 
    
    First, if agent $i$ has utility $0$, letting agent $i$ steal a complementary set of $k$ items from some agents increases their utility to $1$, thereby decreasing the number of $0$ utility agents. 
    The agents that agent $i$ steals from have a utility of strictly more than $\frac{\alpha-1}{k}(v_i(X_i)+\beta) \ge \frac{\beta(\alpha-1)}{k}\ge k$, and a steal will decrease their utility by at most $k$; thus, their utility is still positive after the steal, and the Nash criterion improves. 
    We also note that \Cref{alg:nash-welfare-local-search} never makes an agent who has a positive utility at some step reduce their utility to $0$ at a later step.

     We let the allocation before agent $i$ steals the items be $X$ and the allocation after agent $i$ steals the items be $Y$. By our previous observation, the allocation $Y$ will have weakly fewer agents with utility $0$ than $X$. We assume with only a minor loss of generality that under both allocations all agents have a positive utility. 
     
     Suppose agent $i$ has positive utility. Let $S$ be the set of agents from whom agent $i$ is stealing. Note that $|S|\le k$. Therefore, for all agents $j\in S$,
    \[
    v_j(X_j) > \frac{\alpha-1}{k}(v_i(X_i)+\beta) = \frac{\alpha-1}{k}v_i(X_i) + \frac{(\alpha-1)\beta}{k}\ge k(v_i(X_i)+1)
    \]
    As agents have integer-valued valuations, we get $v_j(X_j)\ge k(v_i(X_i)+1) + 1$. For every agent $j\in S$, $d_j$ denotes the utility reduction due to agent $i$ stealing an item. 
    Note that $\sum_{j\in S}d_j\le k$.

    \begin{align*}
        \frac{\NSW(Y)}{\NSW(X)} &= \left(\frac{(v_i(X_i)+1)}{v_i(X_i)}\prod_{j\in S}\left(1-\frac{d_j}{v_j(X_j)}\right)\right)^{1/n}\\
        &\ge \left(\frac{(v_i(X_i)+1)}{v_i(X_i)}\left(1-\sum_j \frac{d_j}{v_j(X_j)}\right)\right)^{1/n}\\
        &\ge \left(\frac{(v_i(X_i)+1)}{v_i(X_i)}\left(1-\frac{k}{k(v_i(X_i)+1)+1}\right)\right)^{1/n}\\
        &= \left(1 + \frac{1}{v_i(X_i)(k(v_i(X_i)+1)+1)}\right)^{1/n}\\
        &\ge \left(1 + \frac{1}{m^4}\right)^{1/n}
    \end{align*}
The first inequality holds because $\prod_j(1 - a_j)\ge (1 - \sum_j a_j)$ whenever $a_j\in[0,1]$.
\end{proof}

\corNashWelfarePoly*
\begin{proof}
Since there are $m$ items, the total utilitarian welfare of any allocation is at most $\frac mk$. 
By the AM-GM inequality, the Nash welfare of any allocation is at most $m/(kn)$.

By \cref{lem:Nash-increase}, each iteration of the while loop in \cref{alg:nash-welfare-local-search} increases the Nash social welfare by a factor of at least $\left(1+\frac{1}{m^4}\right)^{1/n}$.

Therefore, after at most $nm^4$ iterations, the Nash social welfare increases by a factor of at least $\left(1+\frac{1}{m^4}\right)^{m^4}
\ge 2$.

Since the Nash welfare is upper bounded by $\frac{m}{kn}$, the \textbf{while} loop can execute only a polynomial number of times. Therefore, \cref{alg:nash-welfare-local-search} terminates in polynomial time.
\end{proof}

\subsection{\texorpdfstring{$\frac{1}{k}$-\maxUSW Computation for $\SPOC_{k}$ Preferences}{1/k-maxUSW Computation for SPOC-k Preferences}}

We define a matroid $M=(S,\mathcal{I})$, whose ground set is 
\[
S = \{(i,\{g_1,\ldots,g_k\})\mid \{g_1,\ldots,g_k\} \text{ is a $k$-complement for agent $i$}\}
\]
Intuitively, a tuple $(i,\{g_1,\ldots,g_k\})$ indicates that the $k$-complement $\{g_1,\ldots,g_k\}$ is assigned to agent $i$. Therefore, for any set $I\subseteq S$, we define the rank of $I$ as follows:
\[
r(I) = \sum_i v_i(\{(i,\{g_1,\ldots,g_k\})\mid(i,\{g_1,\ldots,g_k\})\in I\})
\]
Since $r(\cdot)$ is a binary submodular function with $r(\emptyset) =0$, it is a valid rank function for some matroid $M$ over $S$ \cite{oxley2011matroids}. 

Given an item $g\in G$, we define
\[
S_g=\{(i,\{g_1,\ldots,g_k\})\in S \mid g\in\{g_1,\ldots,g_k\}\}.
\]
With these definitions, we define the following \textit{linear program} (LP) named $\LP(M)$:
\begin{align}
\max \quad & \sum_{e \in S} x_e \\
\text{s.t.} \quad 
& \sum_{e \in T} x_e \le r(T) & \forall\, T \subseteq S \label{rank-cons-l}\\
&\sum_{e \in S_g} x_e \le 1 & \forall\, g \in G \label{matching-cons-l}\\
& 0\le  x_e \le 1 & \forall\, e \in S \label{boundary-cons-l}
\end{align}
We first show that every feasible integral solution $\vec x$ to $\LP(M)$ induces an allocation $X$ whose $\USW$ is the objective value of $\vec x$. The proof of \Cref{lm:lp-1-l} is similar to \cref{lm:lp-1}
\begin{restatable}{lemma}{lemLPl}\label{lm:lp-1-l}
For any clean allocation $X=(X_1,\ldots,X_n)$, there exists a corresponding feasible solution to the linear program $\LP(M)$ whose value is $\USW(X)$. Conversely, any integral feasible solution to the LP corresponds to an allocation $X$ such that $\USW(X)$ equals the objective value of the LP.
\end{restatable}
We use \emph{iterative rounding} to find an integral solution to $\LP(M)$ whose value is at least $\frac1k$ of the optimal value.
We say that a solution $\vec x$ to $\LP(M)$ is an \emph{extreme point solution} if it is a vertex of the feasible polytope induced by the LP constraints. 

\begin{algorithm}[H]
\caption{Iterative Rounding}
\label{alg:iterative-matroid-basis-l}
\begin{algorithmic}[1]
\STATE Initialize $T \gets \emptyset$
\WHILE{$S\neq\emptyset$}
    \STATE Find an optimal extreme point solution $\vec x$ to $\mathrm{LP}(M)$
    \IF{ $\exists e\in S$ with $x_e = 0$}
    \STATE $M \gets M \setminus e$
    \ELSIF{there exists an element $e$ such that $x_e \ge \frac{1}{k}$}
        \STATE Update $T \gets T \cup \{e\}$
        \STATE Delete all the edges that share some good with $e$
        \STATE Contract $e$ in $M$, i.e., $M \gets M / e$
    \ENDIF
\ENDWHILE
\RETURN $T$
\end{algorithmic}
\end{algorithm}

\begin{lemma}\label{lm:extreme-point-l}
Let $\vec x$ be an extreme point solution obtained during an iteration of Algorithm~\ref{alg:iterative-matroid-basis-l}. Then there exists an element $e \in S$ such that either $x_e = 0$ or $x_e \ge \tfrac{1}{k}$.
\end{lemma}
\begin{proof}
   Suppose for contradiction that for every element $e \in S$, $x_e \in \left(0, \tfrac{1}{k}\right)$. 
   Let $A$, $B$, and $C$ denote the sets of the tight constraints corresponding to \eqref{rank-cons-l}, \eqref{matching-cons-l}, and \eqref{boundary-cons-l}, respectively. 
   Since $\vec x$ is an extreme point solution, these tight constraints are linearly independent and satisfy
\[
|A| + |B| + |C| = |S|.
\]
Moreover, because $x_e \in (0, \tfrac{1}{k})$ for all $e \in S$, none of the boundary constraints in \eqref{boundary-cons-l} are tight, implying $|C| = 0$. Consequently,
\[
|A| + |B| = |S|.
\]
\begin{enumerate}[label=\textbf{Case \Roman*:},align=left,wide]
\item Suppose that $|A| > \tfrac{|S|}{k+1}$. The tight rank constraints in $A$ form a \emph{chain}; that is, there exists a sequence of subsets of elements of $S$, $\mathcal{L} = \{\chi_0, \chi_1, \ldots, \chi_{|A|}\}$ such that $\emptyset = \chi_0 \subset \chi_1 \subset \cdots \subset \chi_{|A|}$, where $\chi_j$ satisfies $\sum_{e \in \chi_j}x_e = r(\chi_j)$ \cite{schrijver2003combopt}.

Suppose, for the sake of contradiction, that for every $t \in [|A|-1]$ we have
$|\chi_{t+1} \setminus \chi_t| \ge k+1$.
Then
\[
|\chi_{|A|}| = \sum_{t=0}^{|A|-1} |\chi_{t+1} \setminus \chi_t|
\ge (k+1)|A| > |S|,
\]
which contradicts the fact that $\chi_{|A|} \subseteq S$. 
Hence, there exists some $t \in [|A|-1]$ such that $|\chi_{t+1} \setminus \chi_t| \le t$.
Since $r(\chi_{t+1})\ge r(\chi_{t})+1$ and the constraints are tight for $\chi_t$ and $\chi_{t+1}$, this implies that
$\sum_{e \in \chi_{t+1}}x_e \ge \sum_{e \in \chi_t}x_e + 1$. Since the number of summands on the left-hand side is at most $k$ greater than the number of summands on the right-hand side, 
there exists an element $e \in \chi_{t+1} \setminus \chi_t$ with $x_e \ge \tfrac{1}{k}$, a contradiction.
\item Now suppose that $|A| < \tfrac{|S|}{k+1}$, which implies that $|B| > \tfrac{k|S|}{k+1}$. 
Assume, for the sake of contradiction, that for every good $g$ such that $S_g \in B$, we have $|S_g| \ge k+1$. 
Then,
\[
\sum_{g} |S_g| \ge\sum_{g : S_g \in B} |S_g| \ge (k+1)|B| > k|S|.
\]
For any element $e = (i,\{g_1,\ldots,g_k\}) \in S$, the element $e$ appears in exactly $k$ sets, namely $S_{g_1},\ldots, S_{g_k}$. Consequently, $\sum_{g} |S_g| \le k|S|$,
which contradicts the inequality above. Hence, there exists a good $g$ such that $S_g \in B$ and $|S_g| \le k$. It follows that there exists an element $e \in S_g$ with $x_e \ge \tfrac{1}{k}$, yielding the desired contradiction.
\item 
Finally, consider the case where $|A| = \tfrac{|S|}{k+1}$ and $|B| = \tfrac{k|S|}{k+1}$. 
Let $\mathcal{L} = \{\chi_0, \chi_1, \ldots, \chi_{|A|}\}$
denote the list of tight rank constraints in $A$, as in Case I. 
If there exists some $t \in [|A|-1]$ such that $|\chi_{t+1} \setminus \chi_t| \le k$, then the desired conclusion follows immediately from the argument in Case~I.
We assume that $|\chi_{t+1} \setminus \chi_t| \ge k+1$ for all $t \in [|A|-1]$.
It follows that
\[
|\chi_{|A|}| = \sum_{t=0}^{|A|-1} |\chi_{t+1} \setminus \chi_t|
\ge (k+1)|A| = |S|.
\]
Thus, the constraint of $\chi_{|A|}$ must be
\begin{align}
\sum_{e \in S} x_e = r(S) \label{eq:Aconstraints-l}
\end{align}
Next, if there exists a good $g \in G$ such that $S_g \in B$ and $|S_g| \le k$, we are again done. 
Thus, assume that for every $g \in G$ with $S_g \in B$, we have $|S_g| \ge k+1$. Then,
\[
\sum_{g : S_g \in B} |S_g| \ge (k+1)|B| = k|S|.
\]
However, as shown in Case~II, $\sum_{g} |S_g| \le k|S|$ and equality must hold.
Thus, if we sum the rows corresponding to the constraints in $B$, we get a vector of all $k$s. The constraint \eqref{eq:Aconstraints-l} is a row of 1s in the constraint matrix, whereas the sum of constraints in $B$ sums to a row of $k$s. 
Thus, the tight constraints indexed by $A \cup B$ are linearly dependent, contradicting the fact that we picked a set of linearly independent constraints $A \cup B$.
\end{enumerate}
\end{proof}
\begin{proof}[Proof of \Cref{thm:matroid-l}]
Suppose $\vec x$ is an extreme point solution of $\LP(M)$. By \Cref{lm:extreme-point-l}, at least one of the \textbf{If} or \textbf{ElseIf} conditions in the while loop of \Cref{alg:iterative-matroid-basis-l} is satisfied in every iteration. If there exists an element $e^* \in S$ with $x_{e^*} = 0$, then deleting $e^*$ from the matroid $M$ does not change the objective value of the LP defined over $M \setminus e^*$.

Now consider the case where there exists an element $e^* = (i,\{g_1,\ldots,g_k\}) \in S$ such that $x_{e^*} \ge \tfrac{1}{k}$. 
We construct an alternative solution $\vec y$ as follows. 
Let $\vec 1_I$ denote the $|S|$-dimensional indicator vector of an independent set $I$, where $\vec{(1_I)}_{e'} = 1$ if $e' \in I$ and $0$ otherwise. 
Since $\vec x$ lies in the matroid polytope, there exist coefficients $p_I \in [0,1]$ for all $I \in \mathcal{I}$ such that
\begin{align*}
    &\vec x = \sum_{I : e^* \in I} p_I \vec 1_I +\; \sum_{I' : e^* \notin I'} p_{I'} \vec 1_{I'}
\end{align*}
and $\sum_{I \in \mathcal{I}} p_I = 1$.
For any independent set $I' \in \mathcal{I}$ with $e^* \notin I'$, there exists an element $e' \in I'$ and an independent set
$I'' = I' \setminus \{e'\} \cup \{e^*\}$ \cite{oxley2011matroids}.

Define a mapping $f : \mathcal{I} \to \mathcal{I}$ that maps each such $I'$ to a corresponding $I''$. Using this mapping, define
\begin{align*}
    &\vec y = \sum_{I : e^* \in I} p_I  \vec 1_I + \sum_{I' : e^* \notin I'} p_{I'} \vec 1_{f(I')}, \\
    &\forall e \in (S_{g_1} \cup \ldots\cup S_{g_k}) \setminus \{e^*\}:y_{e} = 0, y_{e^*}=1
\end{align*}

By construction, $\vec y$ satisfies all rank constraints~\eqref{rank-cons-l} and boundary constraints~\eqref{boundary-cons-l} of $\LP(M)$. Moreover, setting $y_{e} = 0$ for all $e \in \bigcup_{q=1}^k S_{g_q} \setminus \{e^*\}$ ensures that the matching constraints~\eqref{matching-cons-l} are also satisfied.

Let $\OPT$ denote the objective value of $\vec x$ in $\mathrm{LP}(M)$. Define a vector $\vec y'$ by setting $y'_{e^*} = 0$ and keeping all other coordinates identical to those of $\vec y$. 
Then $\vec y'$ is a feasible solution to $\mathrm{LP}(M / e^*)$. The corresponding objective value of $\vec y'$ is at least
\[
\OPT - k(1 - x_{e^*}) - y_{e^*} \ge \OPT - k,
\]
where the inequality follows from $x_{e^*} \ge \tfrac{1}{k}$ and the fact that $y_{e^*} = 1$.

Therefore, whenever \Cref{alg:iterative-matroid-basis-l} performs a matroid contraction, the objective value of the resulting LP decreases by at most $k$, while the objective value contributed by the selected set $T$ increases by $1$. 
It follows that \Cref{alg:iterative-matroid-basis-l} computes a $\frac{1}{k}$-approximate integral solution to $\mathrm{LP}(M)$, which in turn corresponds to a $\frac{1}{k}$-\maxUSW allocation.

\Cref{alg:iterative-matroid-basis-l} runs in a polynomial number of steps. 
It only remains to argue that one can find a solution to $\LP(M)$ in polynomial time. The number of constraints in $\LP(M)$ is exponential; however, if we are given a \emph{separation oracle}, then the LP can be solved in polynomial time \cite{schrijver2003combopt}.
A separation oracle is an algorithm that can decide whether a given candidate solution $\vec x$ satisfies all feasibility constraints, and output a violating constraint otherwise. Checking whether any of the (polynomially many) constraints in \eqref{matching-cons} and \eqref{boundary-cons} are violated can be done in polynomial time. Since $r$ is a matroid rank function, checking whether one of the matroid rank constraints is violated can be done in polynomial time as well (see \cite[Chapter 5]{lau2011iterative} for details). 
\end{proof}

\section{A \texorpdfstring{$6$}{6}-Approximation for Maximum Nash Social Welfare under Binary XOS Valuations}
In this section, we adapt our analysis of the computation of a $(5,1)$-\maxNSW allocation to the algorithm of Barman and Verma~\cite{Barman2021approxmnw}. This improves the approximation guarantee for binary XOS valuations from $288$ to $6$. We refer the reader to Barman and Verma~\cite{Barman2021approxmnw} for a complete description of the algorithm. Let $X=(X_1,\ldots,X_n)$ denote the clean allocation returned by the algorithm, and let $X^*=(X_1^*,\ldots,X_n^*)$ be a clean $\maxNSW$ allocation. Furthermore, let $X_0$ denote the set of unallocated items. For each agent $i\in N$, the algorithm defines
\[
H_i(X):=\left\{j\in N:\lvert X_j\rvert>4\lvert X_i\rvert\right\}
\quad\text{and}\quad
G_i(X):=X_i\cup X_0\cup\bigcup_{j\in H_i(X)}X_j.
\]
Since $X$ is clean, $\lvert X_j\rvert=v_j(X_j)$ for every $j\in N$. Hence, equivalently, $H_i(X)=\left\{j\in N:v_j(X_j)>4v_i(X_i)\right\}$. Their algorithm terminates when $v_i(G_i(X))<2v_i(X_i)$ for every agent $i\in N$. Observe that $v_i(X_i^*\cap G_i(X))\le v_i(G_i(X)) < 2v_i(X_i)$. This gives us the following upper bound:
\[
v_i(X_i^*) = v_i(X_i^*\cap G_i^c(X)) + v_i(X_i^*\cap G_i(X)) < |X_i^* \cap G_i^c(X)| + 2v_i(X_i).
\]
Therefore,
\begin{align*}
    \frac{\NSW(X^*)}{\NSW(X)} &= \left(\prod_i \frac{v_i(X_i^*)}{v_i(X_i)}\right)^{\frac{1}{n}}\\
    &\le \frac{1}{n}\sum_i \frac{v_i(X_i^*)}{v_i(X_i)}\\
    &\le 2 + \frac{1}{n}\sum_i\frac{|X_i^*\cap G_i^c(X)|}{v_i(X_i)}
\end{align*}
where the first inequality comes from AM-GM inequality.
\begin{lemma}\label{lem:XOS}
$\frac{1}{n}\sum_i\frac{|X_i^*\cap G_i^c(X)|}{v_i(X_i)}\le 4$.
\end{lemma}
\begin{proof}
    Let $U_j(X)$ denote the set of agents who receive utility $j$ under allocation $X$. Therefore, the following equality is true.
    
    \[
        \frac{1}{n}\sum_i\frac{|X_i^*\cap G_i^c(X)|}{v_i(X_i)} = \frac{1}{n}\sum_{j = 1}^{m}\frac{1}{j}\sum_{i\in U_j(X)}|X_i^*\cap G_i^c(X)|.
    \]
    Fix any agent $i \in U_j(X)$. Note that $G_i^c(X)$ is precisely the set of all goods allocated to agents whose utility is $\le4j$. Define
\[
L_j(X) = \bigcup_{j' = 1}^{\min\{4j,m\}} \bigcup_{q\in U_{j'}(X)} X_q,
\]
i.e., $L_j(X)$ is the set of all items assigned in $X$ to agents whose utility is at most $\min\{4j, m\}$. Observe that the sets are nested, i.e., $L_1(X)\subseteq L_2(X) \subseteq \cdots \subseteq L_{m}(X)$.
Furthermore,
\[
|L_j(X)| = \sum_{j'=1}^{\min\{4j,m\}} j'|U_{j'}(X)|.
\]

Consider the linear program:
\begin{equation*}
\begin{split}
\text{maximize} \quad & \sum_{j=1}^{m} \frac{x_j}{j} \\
\text{subject to} \quad 
& \sum_{j=1}^{t} x_j \le |L_t(X)|, \quad t = 1, \ldots, m, \\
& x_j \ge 0, \quad j = 1, \ldots, m.
\end{split}
\end{equation*}

\begin{claim}\label{subclm:feasibility-draft-XOS}
    The choice $x_j = \sum_{i\in U_j(X)}|X_i^*\cap G_i^c(X)|$ is feasible for the LP.
\end{claim}
\begin{proof}
    Observe that for any two agents $i$ and $i'$, the sets $X_i^*\cap G_i^c(X)$ and $X_{i'}^*\cap G_{i'}^c(X)$ are disjoint. Therefore, $\sum_{j=1}^{t}\sum_{i\in U_j(X)}|X_i^*\cap G_i^c(X)|$ is bounded by the total number of goods assigned to agents with utility at most $4t$, which is precisely $|L_t(X)|$.
\end{proof}

\begin{claim}\label{subclm:optimality-draft-XOS}
    The solution $x_j^* = |L_j(X)\setminus L_{j-1}(X)|$, with $L_0(X)=\emptyset$, is optimal for the LP.
\end{claim}
\begin{proof}
    We first observe that the solution $\vec x^* = x_1^*,\dots,x_{m+1}^*$ is feasible; since the sets $L_j(X)$ are nested: 
    $$\sum_{j = 1}^t x_j^* = \sum_{j = 1}^t |L_j(X)\setminus L_{j-1}(X)| = |L_t(X)|.$$
    Next, consider any other feasible solution $\vec x = x_1,\dots,x_{m+1}$ to the LP.  
If all constraints are tight, then the solution $\vec x$ must coincide with $\vec{x}^*$. Suppose for contradiction that $\vec x \ne \vec x^*$ is optimal, and let $t_{\min}$ be the smallest index such that
\[
\sum_{j=1}^{t_{\min}} x_j < |L_{t_{\min}}(X)|.
\]
This implies that $x_{t_{\min}} < |L_{t_{\min}}(X)\setminus L_{t_{\min}-1}(X)|$. Let $\ell > t_{\min}$ be the smallest index for which $x_\ell > 0$.
\begin{enumerate}[label = \textbf{Case \Roman*:},align =left]
    \item Suppose no such $\ell$ exists. 
    Then setting $x_{t_{\min}} = |L_{t_{\min}}(X)\setminus L_{t_{\min}-1}(X)|$ 
maintains feasibility of all constraints, while strictly increasing the objective value, a contradiction to the optimality of $\vec x$.
\item Suppose such a $\ell$ exists. Define 
\[
\delta = \min\left\{\, |L_{t_{\min}}(X)\setminus L_{t_{\min}-1}(X)| - x_{t_{\min}},\; x_\ell \right\}.
\]
Now increase $x_{t_{\min}}$ by $\delta$ and decrease $x_\ell$ by $\delta$. This preserves feasibility of all constraints, yet strictly improves the objective value, again yielding a contradiction.
\end{enumerate}
\end{proof}

By Claim \ref{subclm:feasibility-draft-XOS} and \ref{subclm:optimality-draft-XOS}, the LP objective is attained at $\vec x^*$ such that $x_j^* = |L_j(X)\setminus L_{j-1}(X)|$, and this value upper bounds 
$\sum_{j=1}^{m}\sum_{i\in U_j(X)}\frac{1}{j}|X_i^*\cap G_i^c(X)|$. Hence,
\begin{align}
    \sum_{j=1}^{m}\frac{x_j^*}{j}
    &= \sum_{j=1}^{m}\frac{1}{j}|L_j(X)\setminus L_{j-1}(X)|\label{eq:nash-upperbound-3-draft-XOS}
    \end{align}
    We first observe that 
    \begin{align*}
        |L_j(X) \setminus L_{j-1}(X)| = |L_j(X)| - |L_{j-1}(X)| &= \\
        \sum_{j' = 1}^{\min\{4j,m\}}j'|U_{j'}(X)|-
        \sum_{j' = 1}^{\min\{4(j-1),m\}}j'|U_{j'}(X)| &= \\
        \sum_{j' = \min\{4j-3,m\}}^{\min\{4j,m\}}j'|U_{j'}(X)|&
    \end{align*} 
    
    Thus, \Cref{eq:nash-upperbound-3-draft-XOS} equals
    \begin{align*}
    \sum_{j=1}^{m}\frac{1}{j}
       \sum_{j' = \min\{4j-3,m\}}^{\min\{4j,m\}} j'|U_{j'}(X)| &\le 
    \sum_{j=1}^{m+1} 4\sum_{j' = \min\{4j-3,m\}}^{\min\{4j,m\}} |U_{j'}(X)|\le 4n. 
\end{align*}
The last inequality is due to the fact that the sets $U_{j'}(X)$ partition the agent set $N$, and that each set appears exactly once in the summation.
\end{proof}
Lemma \ref{lem:XOS} concludes the proof.

\section{Missing Proofs from \texorpdfstring{\Cref{sec:usw-efx}}{Section 4}}
\lemCleanEFXRuntime*
\begin{proof}
In every iteration, one agent's utility increases by $1$,
while all other agents' utilities remain unchanged. Therefore, the number of iterations of the \textbf{while} loop is $\le \maxUSW$, i.e., $\left\lfloor\frac{|G|}k\right\rfloor$.
\end{proof}

\begin{algorithm}[ht]
\caption{Completing a Partial EF-X Allocation}
\label{alg:complete-efx}
\begin{algorithmic}[1]
    \STATE $j:=0$ and input partial allocation $A^0:=(A_1^0,\ldots,A_n^0)$
    \WHILE{$G\setminus\bigcup_{i\in N}A_i^j\neq\emptyset$}
        \STATE Choose $g\in G\setminus\bigcup_{i\in N}A_i^j$ and set $E:=\emptyset$
        \FOR{$i\in N$}
            \IF{$\exists i'\in N$ such that $v_{i'}(A_{i'}^j)<v_{i'}(A_i^j\cup\{g\})$}
                \STATE $i^*:=\min\left(\arg\min_{i':\,v_{i'}(A_{i'}^j)<v_{i'}(A_i^j\cup\{g\})}v_{i'}(A_{i'}^j)\right)$
                \STATE $E:=E\cup\{(i,i^*)\}$
            \ENDIF
        \ENDFOR
        \STATE Let $H:=(N,E)$
        \IF{$H$ has a sink $s$}
            \STATE $A_s^{j+1}:=A_s^j\cup\{g\}$ and $A_i^{j+1}:=A_i^j$ for every $i\neq s$
        \ELSE
            \STATE Let $i_1\to i_2\to\cdots\to i_\ell\to i_1$ be a directed cycle in $H$
            \STATE Choose $A_{i_2}^{j+1}\subseteq A_{i_1}^j\cup\{g\}$ as the union of $v_{i_2}(A_{i_2}^j)+1$ independent complementary pairs of agent $i_2$
            \FOR{$t=2,\ldots,\ell$}
                \STATE $A_{i_{t+1}}^{j+1}:=A_{i_t}^j$, where $i_{\ell+1}:=i_1$
            \ENDFOR
            \STATE $A_i^{j+1}:=A_i^j$ for every $i\notin\{i_1,\ldots,i_\ell\}$
        \ENDIF
        \STATE $j:=j+1$
    \ENDWHILE
    \RETURN $A^j$
\end{algorithmic}
\end{algorithm}

\thmCompleteEFX*
\begin{proof}
We prove by induction on $j$ that $A^j$ is EF-X and
$v_i(A_i^j)\ge v_i(A_i^0)$ for every $i\in N$. The claim holds for
$j=0$ by assumption.

Suppose that it holds at the beginning of iteration $j$, and let $g$ be
the unallocated item selected by the algorithm. First, suppose that $H$
contains a sink $s$. Since $s$ has no outgoing edge, no agent strictly
prefers $A_s^j\cup\{g\}$ to her current bundle. Hence, assigning $g$
to $s$ preserves EF-X and no agent's utility decreases.

Now suppose that $H$ has no sink. Since every vertex has exactly one
outgoing edge, $H$ contains a directed cycle
$i_1\to i_2\to\cdots\to i_\ell\to i_1$. For every edge $(i,i')\in E$,
we have $v_{i'}(A_i^j\cup\{g\})>v_{i'}(A_{i'}^j)$. Utilities are
integral, and adding one physical item increases utility by at most one, therefore,
$v_{i'}(A_i^j)\ge v_{i'}(A_{i'}^j)$. Thus, every agent on the cycle
weakly prefers the old bundle of her predecessor to her own bundle.

Since $(i_1,i_2)\in E$, we have $v_{i_2}(A_{i_1}^j\cup\{g\})\ge v_{i_2}(A_{i_2}^j)+1$. Therefore, $A_{i_1}^j\cup\{g\}$ contains
$ v_{i_2}(A_{i_2}^j)+1$ independent complementary bundles of agent $i_2$. The algorithm
sets $A_{i_2}^{j+1}$ to the union of such bundles, and hence
$v_{i_2}(A_{i_2}^{j+1})= v_{i_2}(A_{i_2}^j)+1$. Every other agent on the cycle receives the old bundle of her
predecessor and therefore weakly improves. Agents outside the cycle
retain their bundles, while agent $i_2$ gains one unit of utility.

Now we show that the corresponding allocation is EF-X. Note that agent $i_2$ has the least utility among all agents who envy the bundle $A_{i_1}^j\cup\{g\}$. Therefore, whoever envies $A_{i_2}^{j+1}$, their envy disappears after the removal of any item from that bundle. Furthermore, every bundle other than $A_{i_2}^{j+1}$ belonged to the
previous EF-X allocation. Since no agent's utility decreases, these
bundles remain EF-X. Hence, $A^{j+1}$ is EF-X and the induction claim
continues to hold.

Finally, every cycle update increases the utilitarian social
welfare by at least $1$, and the welfare is bounded. Thus, only $\left\lfloor\frac{|G|}2\right\rfloor$
many cycle updates can occur. Every sink update reduces the number of
unallocated items by $1$. Therefore, \Cref{alg:complete-efx} terminates and the \textbf{while} loop iterates at most $|G|$ times.
\end{proof}
\end{document}